\documentclass[reqno]{amsart}
\usepackage{amsmath,amsthm,amssymb,eucal,mathrsfs,fullpage,setspace,bm}
\usepackage{graphicx}
\usepackage[round]{natbib}
\usepackage[all,cmtip]{xy}
\usepackage{lineno}
\usepackage{subfig}

\title{structural symmetry in evolutionary games}
\author{alex mcavoy and christoph hauert}

\theoremstyle{definition}
\newtheorem*{axiom}{Axiom}
\newtheorem{corollary}{Corollary}
\newtheorem{definition}{Definition}
\newtheorem{example}{Example}
\newtheorem{lemma}{Lemma}
\newtheorem{proposition}{Proposition}
\newtheorem{remark}{Remark}
\newtheorem{theorem}{Theorem}

\newtheorem*{repeatedProposition}{Proposition \ref{prop:functionOfBeta}}
\newtheorem*{repeatedSymmetricTheorem}{Theorem \ref{thm:mainSymmetricTheorem}}
\newtheorem*{repeatedAsymmetricTheorem}{Theorem \ref{thm:mainAsymmetricTheorem}}

\begin{document}

\allowdisplaybreaks

\begin{abstract}
In evolutionary game theory, an important measure of a mutant trait (strategy) is its ability to invade and take over an otherwise-monomorphic population. Typically, one quantifies the success of a mutant strategy via the probability that a randomly occurring mutant will fixate in the population. However, in a structured population, this fixation probability may depend on where the mutant arises. Moreover, the fixation probability is just one quantity by which one can measure the success of a mutant; fixation \textit{time}, for instance, is another. We define a notion of homogeneity for evolutionary games that captures what it means for two single-mutant states, i.e. two configurations of a single mutant in an otherwise-monomorphic population, to be ``evolutionarily equivalent" in the sense that \textit{all} measures of evolutionary success are the same for both configurations. Using asymmetric games, we argue that the term ``homogeneous" should apply to the evolutionary process as a whole rather than to just the population structure. For evolutionary matrix games in graph-structured populations, we give precise conditions under which the resulting process is homogeneous. Finally, we show that asymmetric matrix games can be reduced to symmetric games if the population structure possesses a sufficient degree of symmetry.
\end{abstract}

\maketitle

\section{Introduction}

One of the most basic models of evolution in finite populations is the Moran process \citep{moran:MPCPS:1958}. In the Moran process, a population consisting of two types, a mutant type and a wild type, is continually updated via a birth-death process until only one type remains. The mutant and wild types are distinguished by only their reproductive fitness, which is assumed to be an intrinsic property of a player. A mutant type has fitness $r>0$ relative to the wild type (whose fitness relative to itself is $1$), and in each step of the process an individual is selected for reproduction with probability proportional to fitness. Reproduction is clonal, and the offspring of a reproducing individual replaces another member of the population who is chosen for death uniformly at random. Eventually, this population will end up in one of the monomorphic absorbing states: all mutant type or all wild type. In this context, a fundamental metric of the success of the mutant type is its ability to invade and replace a population of wild-type individuals \citep{nowak:BP:2006}.

In a population of size $N$, the probability that a single mutant in a wild-type population will fixate in the Moran process is
\begin{linenomath}
\begin{align}\label{fpMoran}
\rho &= \frac{1-r^{-1}}{1-r^{-N}} .
\end{align}
\end{linenomath}
In this version of a birth-death process, the members of the population are distinguished by only their types; in particular, there is no notion of spatial arrangement, i.e. the population is \textit{well-mixed}. \citet{lieberman:Nature:2005} extend the classical Moran process to graph-structured populations, which are populations with links between the players that indicate who is a neighbor of whom. In this structured version of the Moran process, reproduction happens with probability proportional to fitness, but the offspring of a reproducing individual can replace only a \textit{neighbor} of the parent. Since individuals are now distinguished by both their types (mutant or wild) \textit{and} locations within the population, a natural question is whether or not the fixation probability of a single mutant type depends on where this mutant appears in the population. \citet{lieberman:Nature:2005} show that this fixation probability is independent of the location of the mutant if everyone has the same number of neighbors, i.e. the graph is \textit{regular} \citep{bollobas:CUP:2001}. In fact, remarkably, the fixation probability of a single mutant on a regular graph is the same as that of Eq. (\ref{fpMoran})--an observation first made in a special case by \citet{maruyama:TPB:1974}. This result, known as the Isothermal Theorem, is independent of the number of neighbors the players have (i.e. the \textit{degree} of the graph).

The Moran process is \textit{frequency-independent} in the sense that the fitness of an individual is determined by type and is not influenced by the rest of the population. However, the Moran model can be easily extended to account for frequency-\textit{dependent} fitness. A standard way in which to model frequency-dependence is through evolutionary games \citep{taylor:MB:1978,hofbauer:CUP:1998,nowak:Nature:2004}. In the classical setup, each player in the population has one of two strategies, $A$ and $B$, and receives an aggregate payoff from interacting with the rest of the population. This aggregate payoff is usually calculated from a sequence of pairwise interactions whose payoffs are described by a payoff matrix of the form
\begin{linenomath}
\begin{align}
\bordermatrix{%
 & A & B \cr
A &\ a & \ b \cr
B &\ c & \ d \cr
} .
\end{align}
\end{linenomath}
Each player's aggregate payoff is then translated into fitness and the strategies in the population are updated based on these fitness values. Since a player's payoff depends on the strategies of the other players in the population, so does that player's fitness. Traditionally, this population is assumed to be infinite, in which case the dynamics of the evolutionary game are governed deterministically by the replicator equation of \citet{taylor:MB:1978}. More recently, evolutionary games have been considered in finite populations \citep{nowak:Nature:2004,taylor:BMB:2004}, where the dynamics are no longer deterministic but rather stochastic. In order to restrict who interacts with whom in the population, these populations can also be given structure. Popular types of structured populations are graphs \citep{lieberman:Nature:2005,ohtsuki:Nature:2006,szabo:PR:2007}, sets \citep{tarnita:PNAS:2009}, and demes \citep{taylor:Selection:2000,hauert:JTB:2012}.

We focus here on evolutionary games in graph-structured populations that proceed in discrete time steps. Such processes define discrete-time Markov chains, either with or without absorbing states (depending on mutation rates). Typically, in evolutionary game theory, one starts with a population of players and repeatedly updates the population based on some update rule such as birth-death \citep{nowak:Nature:2004}, death-birth \citep{ohtsuki:Nature:2006,zukewich:PLoSONE:2013}, imitation \citep{ohtsuki:JTB:2006}, pairwise comparison \citep{szabo:PRE:1998,traulsen:JTB:2007}, or Wright-Fisher \citep{ewens:S:2004,imhof:JMB:2006}. These update rules can be split into two classes: cultural and genetic \citep[see][]{mcavoy:PLOSCB:2015}. Cultural update rules involve strategy imitation while genetic update rules involve reproduction and inheritance. Without mutations, an update rule may be seen as giving a probability distribution over a number of strategy-acquisition scenarios: a player inherits a new strategy through imitation (cultural rules) or is born with a strategy determined by the parent(s) (genetic rules). Mutation rates disrupt these scenarios by placing a small probability of a player taking on a novel strategy. The way in which strategy mutation rates are incorporated into an evolutionary process depends on both the class of the update rule and the specifics of the update rule itself. In a general sense, we say that strategy mutations are \textit{homogeneous} if they depend on neither the players themselves nor the locations of the players. This notion of homogeneous strategy mutations is analogous to that of a symmetric game, which is a game for which the payoffs depend on the strategies played but are independent of the identities and locations of the players.

The Isothermal Theorem seems to indicate that populations structured by regular graphs possess a significant degree of homogeneity, meaning that different locations within the population appear to be equivalent for the purposes of evolutionary dynamics. However, it is important to note that (i) fixation probability is just one metric of evolutionary success and (ii) the Moran process is only one example of an evolutionary process. For example, in addition to the probability of fixation, one could look at the \textit{absorption time}, which is the average number of steps until one of the monomorphic absorbing states is reached. Moreover, one could consider frequency-dependent processes, possibly with different update rules, in which fitness is no longer an intrinsic property of an individual but is also influenced by the other members of the population. We show that the Isothermal Theorem does not extend to arbitrary frequency-dependent processes such as evolutionary games. Furthermore, we show that this theorem does not apply to fixation \textit{times}; that is, even for the Moran process on a regular graph, the average number of updates until a monomorphic absorbing state is reached can depend on the initial placement of the mutant.

Given that the Isothermal Theorem does not extend to other processes defined on regular graphs, the next natural question is the following: what is the meaning of a spatially-homogeneous population in evolutionary game theory? In fact, we argue using asymmetric games \citep{mcavoy:PLOSCB:2015} that the term ``homogeneous" should apply to an evolutionary process \textit{as a whole} rather than to just the population structure. Even for populations that appear to be spatially homogeneous, such as populations on complete graphs, non-uniform distribution of resources within the population can result in heterogeneity of the overall process. Similarly, for symmetric games, heterogeneity can be introduced into the dynamics of an evolutionary process through strategy mutations. Therefore, a notion of homogeneity of an evolutionary game should take into account \textit{at least} (i) population structure, (ii) payoffs, and (iii) strategy mutations.

If the strategy-mutation rates are miniscule, then the population spends most of its time in monomorphic states. With small mutation rates, one can define an embedded Markov chain on the monomorphic states and use this chain to study the success of each strategy \citep{fudenberg:JET:2006,wu:JMB:2011}. That is, when a mutation occurs, the population is assumed to return to a monomorphic state before another mutant arises. Thus, the states of interest are the monomorphic states and the states consisting of a single mutant in an otherwise-monomorphic population. We say that an evolutionary game is \textit{homogeneous} if any two states consisting of a single mutant ($A$-player) in wild-type population ($B$-players) are mathematically equivalent. We make precise what we mean by ``mathematically equivalent" in \S\ref{sec:markov}, but, informally, this equivalence means that any two such states are the same up to relabeling. In particular, \textit{all} metrics such as fixation probability, absorption time, etc. are the same for any two states consisting of a single $A$-mutant in a $B$-population. We show that an evolutionary game in a graph-structured population is homogeneous if the graph is vertex-transitive (``looks the same" from each vertex), the payoffs are symmetric, and the strategy mutations are homogeneous. This result holds for any update rule and selection intensity.

Finally, we explore the effects of population structure on asymmetric evolutionary games. In the weak selection limit, we show that asymmetric matrix games with homogeneous strategy mutations can be reduced to symmetric games if the population structure is arc-transitive (``looks the same" from each edge in the graph). This result is a finite-population analogue of the main result of \citet{mcavoy:PLOSCB:2015}, which states that a similar reduction to symmetric games is possible in sufficiently large populations. Thus, we establish that this reduction applies to any population size if the graph possesses a sufficiently high degree of symmetry. Our explorations, both for symmetric and asymmetric games, clearly demonstrate the effects of population structure, payoffs, and strategy mutations on symmetries in evolutionary games.

\section{Markov chains and evolutionary equivalence}\label{sec:markov}

\subsection{General Markov chains}\label{subsec:generalMarkov}

The evolutionary processes we consider here define discrete-time Markov chains on finite state spaces. The notions of symmetry and evolutionary equivalence that we aim to introduce for evolutionary processes can actually be stated quite succinctly at the level of the Markov chain. We first work with general Markov chains, and later we apply these ideas to evolutionary games.

\begin{definition}[Symmetry of states]\label{def:symmetry}
Suppose that $X=\left\{X_{n}\right\}_{n\geqslant 0}$ is a Markov chain on a (finite) state space, $\mathcal{S}$, with transition matrix $\mathbf{T}$. An automorphism of $X$ is a bijection $\phi :\mathcal{S}\rightarrow\mathcal{S}$ such that $\mathbf{T}_{\mathfrak{s},\mathfrak{s}'}=\mathbf{T}_{\phi\left(\mathfrak{s}\right) ,\phi\left(\mathfrak{s}'\right)}$ for each $\mathfrak{s},\mathfrak{s}'\in\mathcal{S}$. Two states $\mathfrak{s},\mathfrak{s}'\in\mathcal{S}$ are said to be \textit{symmetric} if there exists $\phi\in\textrm{Aut}\left(X\right)$ such that $\phi\left(\mathfrak{s}\right) =\mathfrak{s}'$.
\end{definition}

Definition \ref{def:symmetry} says that the states of the chain can be relabeled in such a way that the transition probabilities are preserved. This relabeling may affect the long-run distribution of the chain since it need not fix absorbing states, so we make one further refinement in order to ensure that if two states are symmetric, then they behave in the same way:
\begin{definition}[Evolutionary equivalence]\label{def:evolutionaryEquivalence}
States $\mathfrak{s}$ and $\mathfrak{s}'$ are \textit{evolutionarily equivalent} if there exists an automorphism of the Markov chain, $\phi\in\textrm{Aut}\left(X\right)$, such that
\begin{enumerate}

\item[(i)] $\phi\left(\mathfrak{s}\right) =\mathfrak{s}'$;

\item[(ii)] if $\mu$ is a stationary distribution of $X$, then $\phi\left(\mu\right) =\mu$.

\end{enumerate}
\end{definition} 
For a Markov chain with absorbing states, the notions of symmetry and evolutionary equivalence of states need not coincide (see Example \ref{ex:swapStrategies} of Appendix B). However, if the Markov chain has a unique stationary distribution (as would be the case if it were irreducible), then symmetry implies evolutionary equivalence:
\begin{proposition}\label{prop:ergodicSymmetry}
If $X$ has a unique stationary distribution, then two states are symmetric if and only if they are evolutionarily equivalent.
\end{proposition}
We show in Appendix B that a Markov chain symmetry preserves the set of stationary distributions (Lemma \ref{lem:stationaryDistribution}), so if there is a unique stationary distribution, then condition (ii) of Definition \ref{def:evolutionaryEquivalence} is satisfied automatically by any symmetry. Proposition \ref{prop:ergodicSymmetry} is then an immediate consequence of this result.

If $\mathfrak{s}$ and $\mathfrak{s}'$ are evolutionarily equivalent, then it is clear for absorbing processes that the probability that $\mathfrak{s}$ fixates in absorbing state $\overline{\mathfrak{s}}$ is the same as the probability that $\mathfrak{s}'$ fixates in state $\overline{\mathfrak{s}}$ (and similarly for fixation times). If the process has a unique stationary distribution, then the symmetry of $\mathfrak{s}$ and $\mathfrak{s}'$ implies that this distribution puts the same mass on $\mathfrak{s}$ and $\mathfrak{s}'$. These properties follow at once from the fact that the states $\mathfrak{s}$ and $\mathfrak{s}'$ are equivalent up to relabeling.

\subsection{Markov chains defined by evolutionary games}

Our focus is on evolutionary games on fixed population structures. If $S$ is a finite set of strategies (or ``actions") available to each player, and if the population size is $N$, then the state space of the Markov chain defined by an evolutionary game in such a population is $\mathcal{S}=S^{N}$. For evolutionary games without random strategy mutations, the absorbing states of the chain are the monomorphic states, i.e. the strategy profiles consisting of just a single unique strategy. Thus, states $\mathfrak{s}$ and $\mathfrak{s}'$ are evolutionarily equivalent if they are symmetric relative to the monomorphic states. On the other hand, evolutionary processes with strategy mutations are typically irreducible (and have unique stationary distributions); in these processes, the notions of symmetry and evolutionary equivalence coincide by Proposition \ref{prop:ergodicSymmetry}.

In order to state the definition of a \textit{homogeneous} evolutionary process, we first need some notation. For $s,s'\in S$, denote by $\mathfrak{s}_{\left(s',i\right) ,s}$ the state in $S^{N}$ whose $i$th coordinate is $s'$ and whose $j$th coordinate for $j\neq i$ is $s$; that is, all players are using strategy $s$ except for player $i$, who is using $s'$.
\begin{definition}[Homogeneous evolutionary process]\label{def:homogeneousEvolutionaryProcess}
An evolutionary process on $S^{N}$ is \textit{homogeneous} if for each $s,s'\in S$, the states $\mathfrak{s}_{\left(s',i\right) ,s}$ and $\mathfrak{s}_{\left(s',j\right) ,s}$ are evolutionarily equivalent for each $i,j=1,\dots ,N$. An evolutionary process is \textit{heterogeneous} if it is not homogeneous.
\end{definition}

In other words, an evolutionary process is homogeneous if, at the level of the Markov chain it defines, any two states consisting of a single mutant in an otherwise-monomorphic population appear to be relabelings of one another. As noted in \S\ref{subsec:generalMarkov}, \textit{all} quantities with which one could measure evolutionary success are the same for these single-mutant states if the process is homogeneous.

\section{Evolutionary games on graphs}

We consider evolutionary games in graph-structured populations. Unless indicated otherwise, a ``graph" means a directed, weighted graph on $N$ vertices. A directed graph is one in which the edges have orientations, meaning there may be an edge from $i$ to $j$ but not from $j$ to $i$. Moreover, the edges carry weights, which we assume are nonnegative real numbers. A directed, weighted graph is equivalent to a nonnegative $N\times N$ matrix, $\mathscr{D}$, where there is an edge from $i$ to $j$ if and only if $\mathscr{D}_{ij}\neq 0$. If there is such an edge, then the weight of this edge is simply $\mathscr{D}_{ij}$. Since there is a one-to-one correspondence between directed, weighted graphs on $N$ vertices and $N\times N$ real matrices, we refer to graphs and matrices using the same notation, describing $\mathscr{D}$ as a graph but using the matrix notation $\mathscr{D}_{ij}$ to indicate the weight of the edge from vertex $i$ to vertex $j$.

Every graph considered here is assumed to be connected (strongly), which means that for any two vertices, $i$ and $j$, there is a (directed) path from $i$ to $j$. This assumption is not that restrictive in evolutionary game theory since one can always partition a graph into its strongly-connected components and study the behavior of an evolutionary process on each of these components. Moreover, for evolutionary processes on graphs that are not strongly connected, it is possible to have both (i) recurrent non-monomorphic states in processes without mutations and (ii) multiple stationary distributions in processes with mutations. Some processes (such as the death-birth process) may not even be defined on graphs that are not strongly connected. Therefore, we focus on strongly-connected graphs and make no further mention of the term ``connected."

Since our goal is to discuss symmetry in the context of evolutionary processes, we first describe several notions of symmetry for graphs. The three types of graphs we treat here are \textit{regular}, \textit{vertex-transitive}, and \textit{symmetric}. Informally speaking, a regular graph is one in which each vertex has the same number of neighboring vertices (and this number is known as the \textit{degree} of the graph). A vertex-transitive graph is one that looks the same from any two vertices; based on the graph structure alone, a player cannot tell if he or she has been moved from one location to another. A symmetric (or \textit{arc-transitive}) graph is one that looks the same from any two edges. That is, if two players are neighbors and are both moved to another pair of neighboring vertices, then they cannot tell that they have been moved based on the structure of the graph alone. We recall in detail the formal definitions of these terms in Appendix B. The relationships between these three notions of symmetry, as well as some examples, are illustrated in Fig. \ref{fig:graphDiagram}.

\begin{figure}
\begin{center}
\includegraphics[scale=0.15]{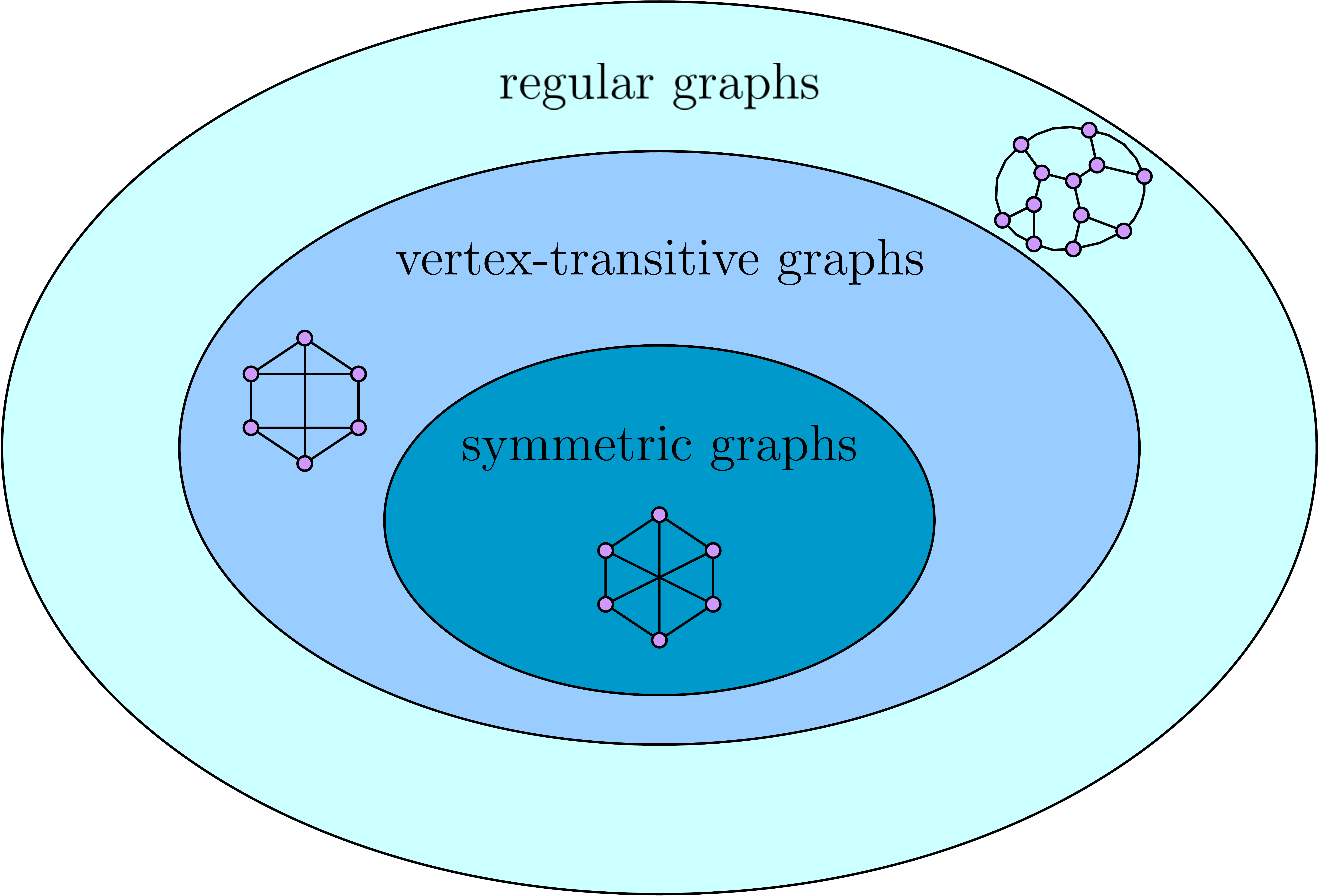}
\caption{Three different levels of symmetry for connected graphs. Regular graphs have the property that the degrees of the vertices are all the same. Vertex-transitive graphs look the same from each vertex and are necessarily regular. Symmetric (arc-transitive) graphs look the same from any two (directed) edges. Each of these containments is strict; there exist graphs that are regular but not vertex-transitive (Fig. \ref{fig:fruchtGraph}) and vertex-transitive but not symmetric (Fig. \ref{fig:transitiveSymmetric}(A)).\label{fig:graphDiagram}}
\end{center}
\end{figure}

We turn now to evolutionary processes in graph-structured populations:
\subsection{The Moran process}

Consider the Moran process on a graph, $\mathscr{D}$. \citet{lieberman:Nature:2005} show that if $\mathscr{D}$ is regular, then the fixation probability of a randomly placed mutant is given by Eq. (\ref{fpMoran}), the fixation probability of a single mutant in the classical Moran process. This result (known as the Isothermal Theorem) proves that, in particular, this fixation probability does not depend on the initial location of the mutant. (We refer to this latter statement as the ``weak" version of the Isothermal Theorem.) Our definition of homogeneity in the context of evolutionary processes (Definition \ref{def:homogeneousEvolutionaryProcess}) is related to this independence of initial location and has nothing to do with fixation probabilities in the classical Moran process. Naturally, the Isothermal Theorem raises the question of whether or not this location independence extends to absorption times (average number of steps until an absorbing state is reached) when $\mathscr{D}$ is regular.

Suppose that $\mathscr{D}$ is the \textit{Frucht graph} of Fig. \ref{fig:fruchtGraph}. The Frucht graph is an undirected, unweighted, regular (but not vertex-transitive) graph of size $12$ and degree $3$ \citep{frucht:CM:1939}. The fixation probabilities and absorption times of a single mutant in a wild-type population are given in Fig. \ref{fig:moranBarGraphs} as a function of the initial location of the mutant. The fixation probabilities do not depend on the initial location of the mutant, as predicted by the Isothermal Theorem, but the absorption times \textit{do} depend on where the mutant arises. In fact, the absorption time is distinct for each different initial location of the mutant. The details of these calculations are in Appendix C. Therefore, even the weak form of the Isothermal Theorem fails to hold for absorption times. In particular, the Moran process on a regular graph need not define a homogeneous evolutionary process.

\begin{figure}
\begin{center}
\includegraphics[scale=0.5]{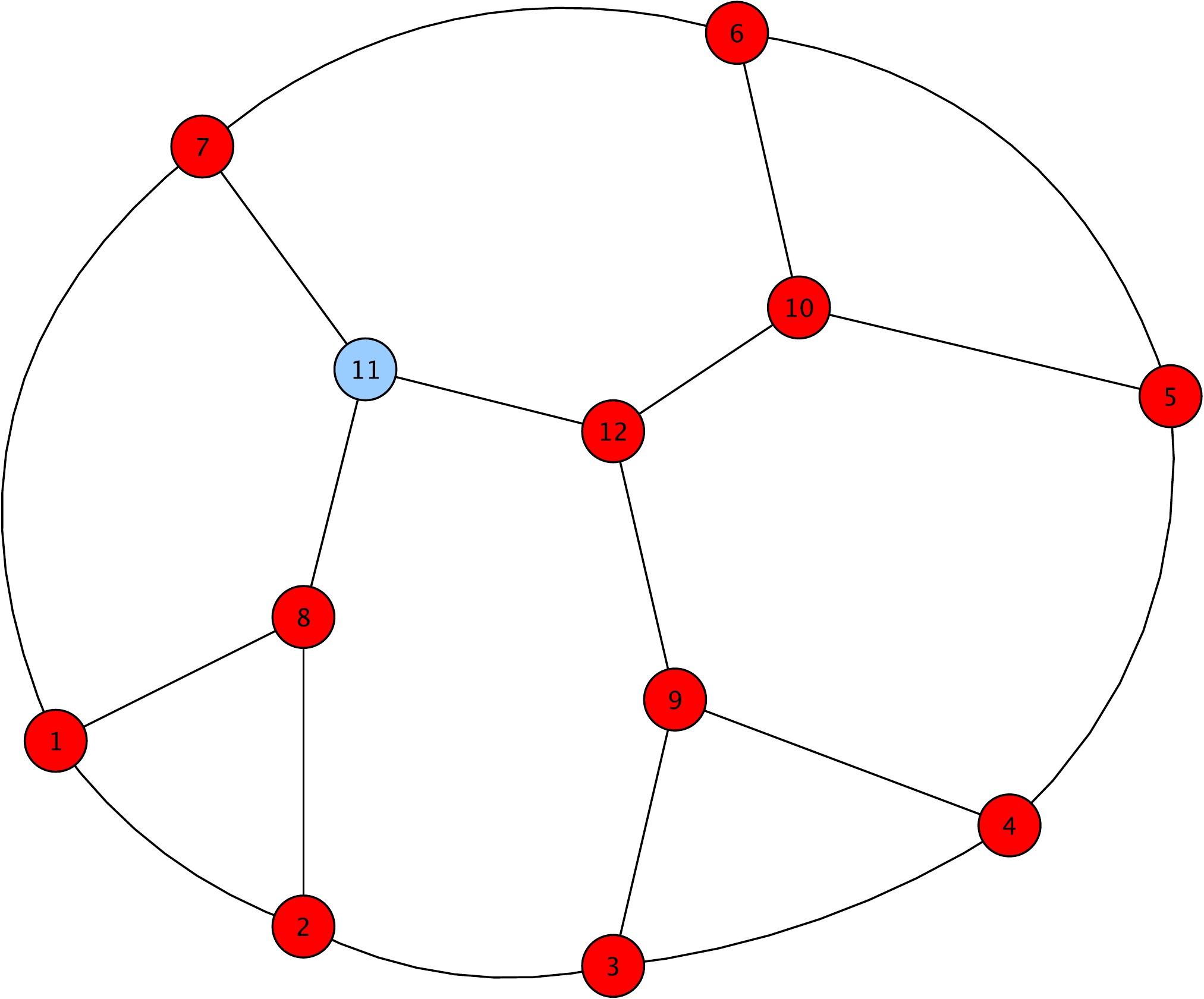}
\caption{A single mutant (cooperator) at vertex $11$ of the Frucht graph. In the Snowdrift Game, the probability that cooperators fixate depends on the initial location of this mutant on the Frucht graph (even if the intensity of selection is weak).\label{fig:fruchtGraph}}
\end{center}
\end{figure}

\begin{figure}
\begin{center}
\subfloat[]{\includegraphics[scale=0.45]{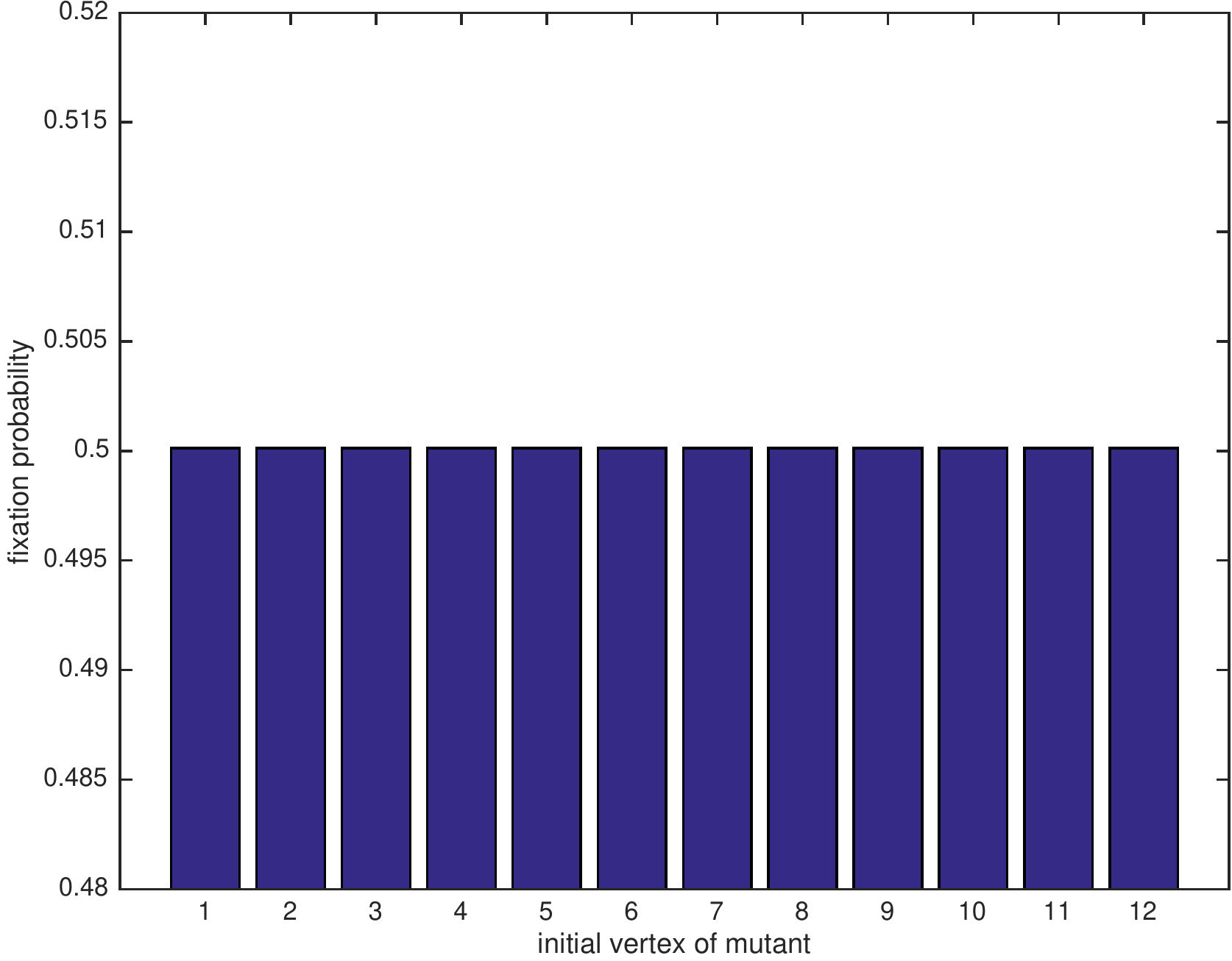}}
\quad\quad
\subfloat[]{\includegraphics[scale=0.45]{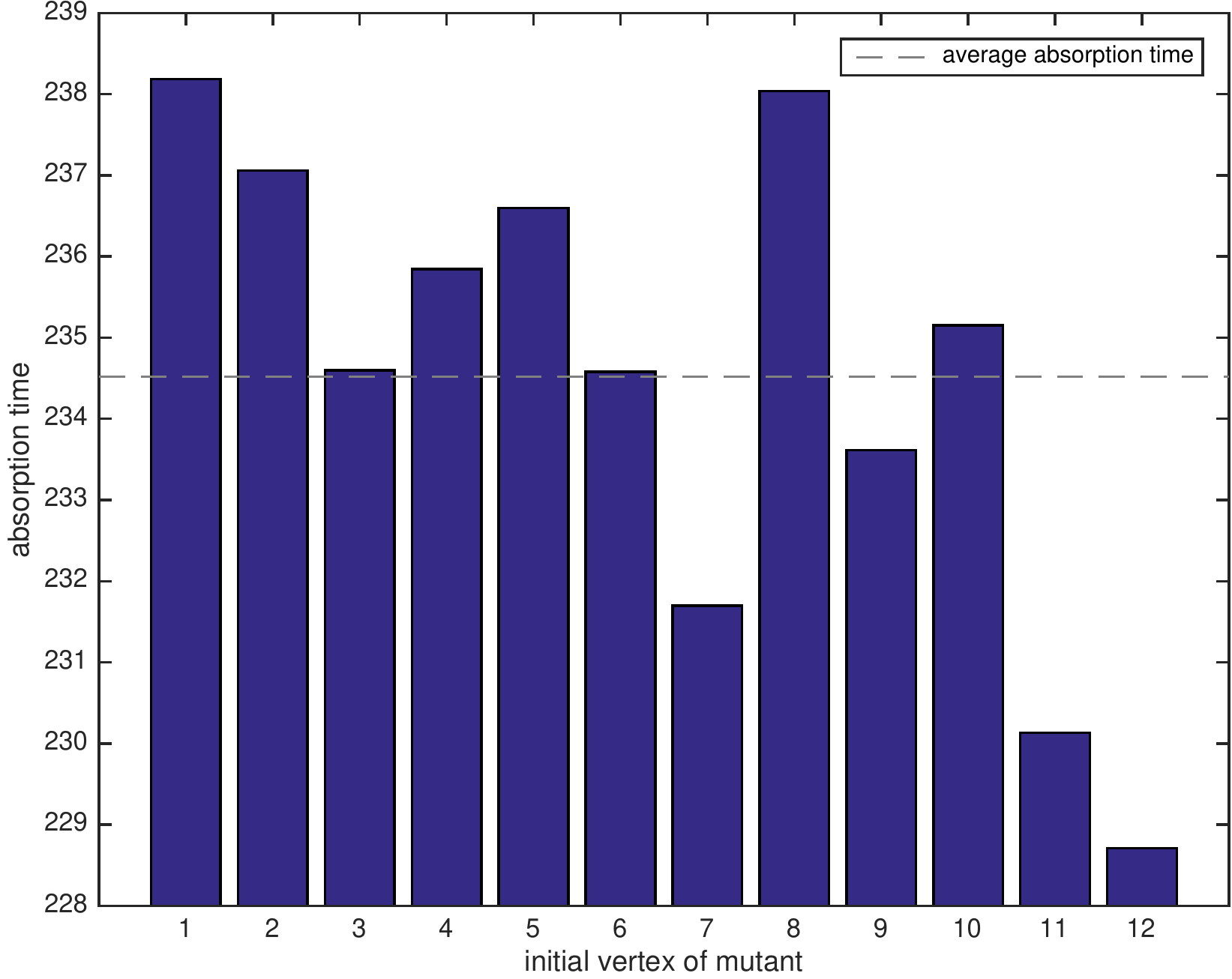}}
\end{center}
\caption{Fixation probability (A) and absorption time (B) versus initial vertex of mutant for the Moran process on the Frucht graph. In both figures, the mutant has fitness $r=2$ relative to the wild type. As predicted by the Isothermal Theorem, the fixation probability does not depend on the initial location of the mutant. The absorption time (measured in number of updates) is different for each initial placement of the mutant on the Frucht graph. The precise values of these fixation probabilities and absorption times are in Appendix C. \label{fig:moranBarGraphs}}
\end{figure}

This setup involving two types of players, frequency-independent interactions, and a population structure defined by a single graph, can be generalized considerably:

\subsection{Symmetric games}\label{subsec:frequencyDependentGames}

A powerful version of evolutionary graph theory uses \textit{two} graphs to define relationships between the players: an \textit{interaction graph}, $\mathscr{E}$, and a \textit{dispersal graph}, $\mathscr{D}$ \citep{ohtsuki:PRL:2007,taylor:Nature:2007,ohtsuki:JTB:2007,pacheco:PLoSCB:2009,debarre:NC:2014}. These graphs both have nonnegative weights. As an example of how these two graphs are used to define an evolutionary process, we consider a birth-death process based on two-player, symmetric interactions:
\begin{example}\label{interactionDispersalExample}
Consider a symmetric matrix game with $n$ strategies, $A_{1},\dots ,A_{n}$, and payoff matrix
\begin{linenomath}
\begin{align}\label{generic2by2}
\bordermatrix{%
 & A_{1} & A_{2} & \cdots & A_{n} \cr
A_{1} &\ a_{11} &\ a_{12} &\ \cdots &\ a_{1n} \cr
A_{2} &\ a_{21} &\ a_{22} &\ \cdots &\ a_{2n} \cr
\ \vdots &\ \vdots &\ \vdots &\ \ddots &\ \vdots \cr
A_{n} &\ a_{n1} &\ a_{n2} &\ \cdots &\ a_{nn} \cr
} .
\end{align}
\end{linenomath}
If $\left(s_{1},\dots ,s_{N}\right)\in\left\{1,\dots ,n\right\}^{N}$, then the total payoff to player $i$ is
\begin{linenomath}
\begin{align}\label{eq:symmetricMatrixPayoff}
u_{i}\left(s_{1},\dots ,s_{N}\right) &:= \sum_{j=1}^{N}\mathscr{E}_{ij}a_{s_{i}s_{j}} .
\end{align}
\end{linenomath}
If $\beta\geqslant 0$ is the intensity of selection, then the fitness of player $i$ is
\begin{linenomath}
\begin{align}
f_{\beta}\Big(u_{i}\left(s_{1},\dots ,s_{N}\right)\Big) &:= \exp\Big\{\beta u_{i}\left(s_{1},\dots ,s_{N}\right)\Big\} .
\end{align}
\end{linenomath}
In each time step, a player (say, player $i$) is chosen for reproduction with probability proportional to fitness. With probability $\varepsilon\geqslant 0$, the offspring of this player adopts a novel strategy uniformly at random from $\left\{A_{1},\dots ,A_{n}\right\}$; with probability $1-\varepsilon$, the offspring inherits the strategy of the parent. Next, another member of the population is chosen for death, with the probability of player $j$ dying proportional to $\mathscr{D}_{ij}$. The offspring then fills the vacancy created by the deceased neighbor, and the process repeats. $\mathscr{E}$ is called the ``interaction" graph since it governs payoffs based on encounters, and $\mathscr{D}$ is called the ``dispersal" graph since it is involved in strategy propagation.
\end{example}

\subsubsection{Heterogeneous evolutionary games}

We now explore the ways in which population structure and strategy mutations can introduce heterogeneity into an evolutionary process. Consider the Snowdrift Game with strategies $C$ (cooperate) and $D$ (defect) and payoff matrix
\begin{linenomath}
\begin{align}\label{snowdriftGame}
\bordermatrix{%
& C & D \cr
C &\ 5 & \ 3 \cr
D &\ 7 & \ 0 \cr
} .
\end{align}
\end{linenomath}
Suppose that $\mathscr{E}$ and $\mathscr{D}$ are both the (undirected, unweighted) Frucht graph (see Fig. \ref{fig:fruchtGraph}). If the selection intensity is $\beta =1$, then the fixation probability of a single cooperator in a population of defectors in a death-birth process depends on the initial location of the cooperator (Fig. \ref{fig:snowdriftBarGraphs}). Since the Frucht graph is regular (but not vertex-transitive), this example demonstrates that the Isothermal Theorem does not extend to frequency-dependent games. In particular, symmetric games on regular graphs can be heterogeneous, and regularity of the graph does not imply that the ``fixation probability of a randomly placed mutant" is well-defined. This dependence of the fixation probability on the initial location of the mutant is not specific to the Snowdrift Game or the death-birth update rule; one can show that it also holds for the Donation Game in place of the Snowdrift Game or the birth-death rule in place of the death-birth rule, for instance.

\begin{figure}
\begin{center}
\subfloat[]{\includegraphics[scale=0.45]{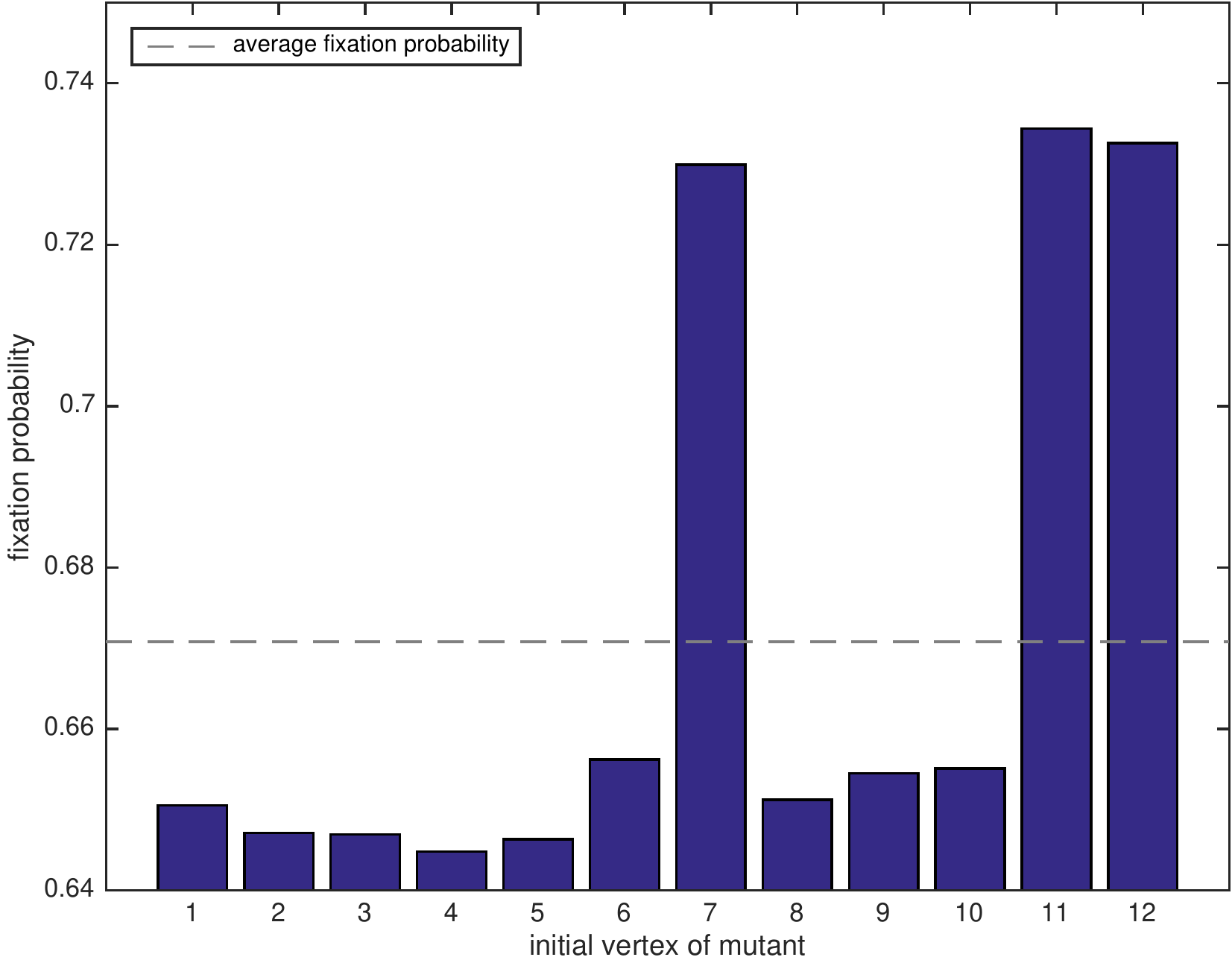}}
\quad\quad
\subfloat[]{\includegraphics[scale=0.45]{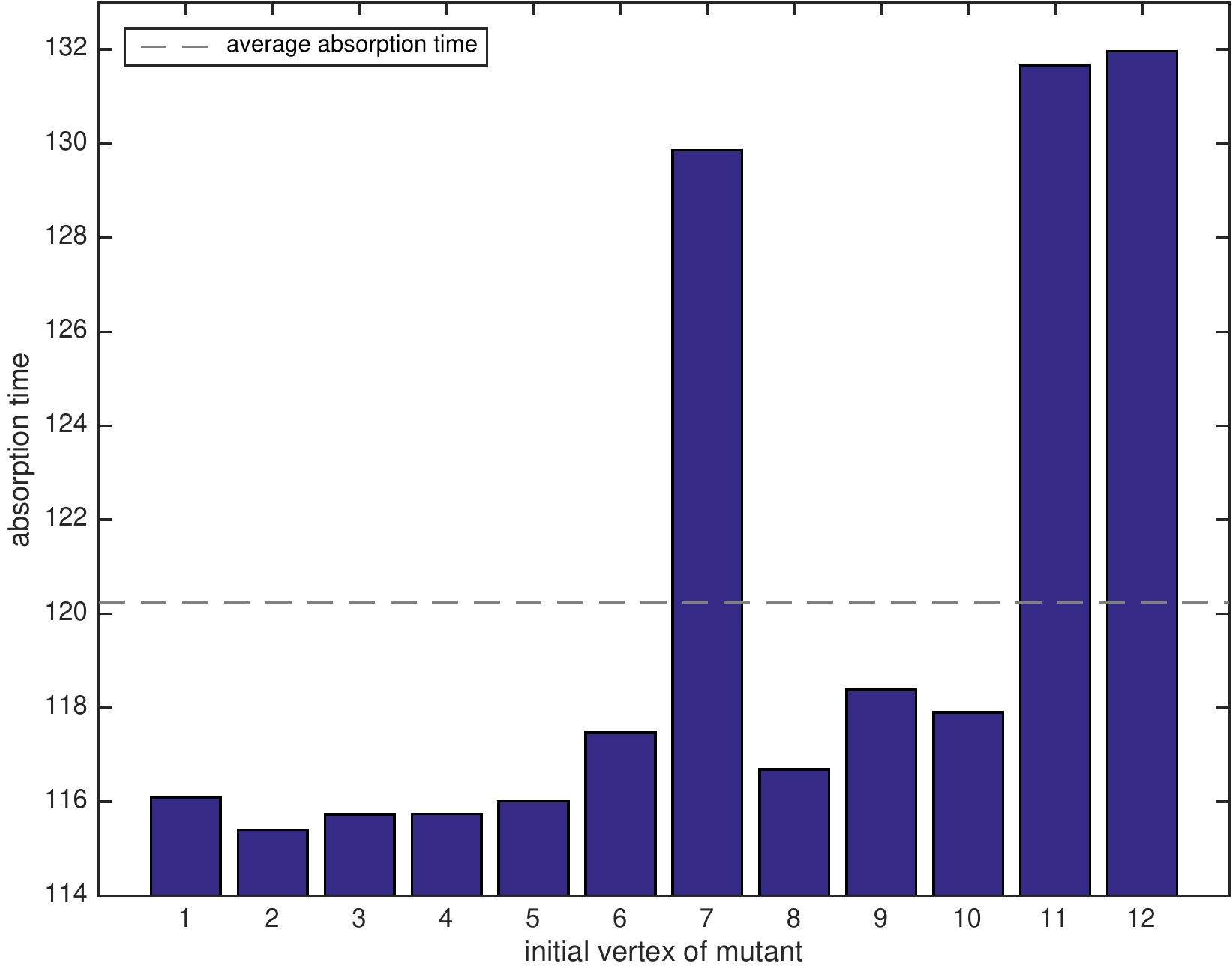}}
\end{center}
\caption{Fixation probability (A) and absorption time (B) versus initial vertex of mutant (cooperator) for a death-birth process on the Frucht graph. In both figures, the game is a Snowdrift Game whose payoffs are described by payoff matrix (\ref{snowdriftGame}), and the selection intensity is $\beta =1$. Unlike the Moran process, this process is frequency-dependent, and it is evident that \textit{both} fixation probabilities and absorption times (measured in number of updates) depend on the initial placement of the mutant. See Appendix C for details. Notably, the three vertices ($7$, $11$, and $12$) for which both fixation probability and absorption time are highest are the only vertices in the Frucht graph not appearing in a three-cycle.\label{fig:snowdriftBarGraphs}}
\end{figure}

With $\beta =1$, the selection intensity is fairly strong, which raises the question of whether or not these fixation probabilities still differ if selection is weak. In fact, our observation for this value of $\beta$ is not an anomaly: Suppose that $\mathfrak{s}$ and $\mathfrak{s}'$ are states (indicating some non-monomorphic initial configuration of strategies), and that $\mathfrak{s}_{i}$ and $\mathfrak{s}_{j}$ are monomorphic absorbing states (indicating states in which each player uses the same strategy). Let $\rho_{\mathfrak{s},i}$ denote the probability that state $i$ is reached after starting in state $\mathfrak{s}$, and let $\mathbf{t}_{\mathfrak{s}}$ denote the average number of updates required for the process to reach an absorbing state after starting in state $\mathfrak{s}$. Each of $\rho_{\mathfrak{s},i}$ and $\mathbf{t}_{\mathfrak{s}}$ may be viewed as functions of $\beta$, and we have the following result:
\begin{proposition}\label{prop:functionOfBeta}
Each of the equalities
\begin{linenomath}
\begin{subequations}
\begin{align}
\rho_{\mathfrak{s},i} &= \rho_{\mathfrak{s}',j} ; \\
\mathbf{t}_{\mathfrak{s}} &= \mathbf{t}_{\mathfrak{s}'}
\end{align}
\end{subequations}
\end{linenomath}
holds for either (i) \textit{every} $\beta\geqslant 0$ or (ii) at most finitely many $\beta\geqslant 0$. Thus, if one of these equalities fails to hold for even a single value of $\beta$, then it fails to hold for all sufficiently small $\beta >0$.
\end{proposition}
For a proof of Proposition \ref{prop:functionOfBeta}, see Appendix A. This result allows one to conclude that if there are differences in fixation probabilities or times for large values of $\beta$ (where these differences are more apparent), then there are corresponding differences in the limit of weak selection.

Even if a symmetric game is played in a well-mixed population, heterogeneous strategy mutations may result in heterogeneity of the evolutionary process. Consider, for example, the pairwise comparison process \citep{szabo:PRE:1998,traulsen:JTB:2007} based on the symmetric Snowdrift Game, (\ref{snowdriftGame}), in a well-mixed population with $N=3$ players. We model this well-mixed population using a complete, undirected, unweighted graph of size $3$ for each of $\mathscr{E}$ and $\mathscr{D}$ (see Fig. \ref{fig:wellMixedThree}). For $i\in\left\{1,2,3\right\}$, let $\varepsilon_{i}\in\left[0,1\right]$ be the strategy-mutation (``exploration") rate for player $i$. These strategy mutations are incorporated into the process as follows: At each time step, a focal player (player $i$) is chosen uniformly at random to update his or her strategy. A neighbor (one of the two remaining players) is then chosen randomly as a model player. If $\beta$ is the selection intensity, $\pi_{\textrm{f}}$ is the payoff of the focal player, and $\pi_{\textrm{m}}$ is the payoff of the model player, then the focal player imitates the strategy of the model player with probability
\begin{linenomath}
\begin{align}
\frac{1-\varepsilon_{i}}{1+e^{-\beta\left(\pi_{\textrm{m}}-\pi_{\textrm{f}}\right)}} .
\end{align}
\end{linenomath}
and chooses to retain his or her strategy with probability
\begin{linenomath}
\begin{align}
\frac{1-\varepsilon_{i}}{1+e^{-\beta\left(\pi_{\textrm{f}}-\pi_{\textrm{m}}\right)}} .
\end{align}
\end{linenomath}
With probability $\varepsilon_{i}$, the focal player adopts a new strategy uniformly at random from the set $\left\{C,D\right\}$, irrespective of the current strategy. Provided at least one of $\varepsilon_{1}$, $\varepsilon_{2}$, and $\varepsilon_{3}$ is positive, the Markov chain on $\left\{C,D\right\}^{3}$ defined by this process is irreducible and has a unique stationary distribution, $\mu$. Let $\varepsilon_{1}=0.01$ and $\varepsilon_{2}=\varepsilon_{3}=0$. Since the mutation rate depends on the location, $i$, the strategy mutations are \textit{heterogeneous}. If the selection intensity is $\beta =1$, then a direct calculation (to four significant figures) gives
\begin{linenomath}
\begin{align}
\mu\left(C,D,D\right) &= 0.005812 \neq 0.0004897 = \mu\left(D,C,D\right) ,
\end{align}
\end{linenomath}
where $\mu\left(C,D,D\right)$ (resp. $\mu\left(D,C,D\right)$) is the mass $\mu$ places on the state $\left(C,D,D\right)$ (resp. $\left(D,C,D\right)$). Therefore, by Proposition \ref{prop:ergodicSymmetry} and Definition \ref{def:homogeneousEvolutionaryProcess}, this evolutionary process is \textit{not} homogeneous, despite the fact that the population is well-mixed and the game is symmetric. This result is not particularly surprising, but it clearly illustrates the effects of hetergeneous strategy-mutation rates on symmetries of the overall process.

\begin{figure}
\begin{center}
\subfloat[]{\includegraphics[scale=0.45]{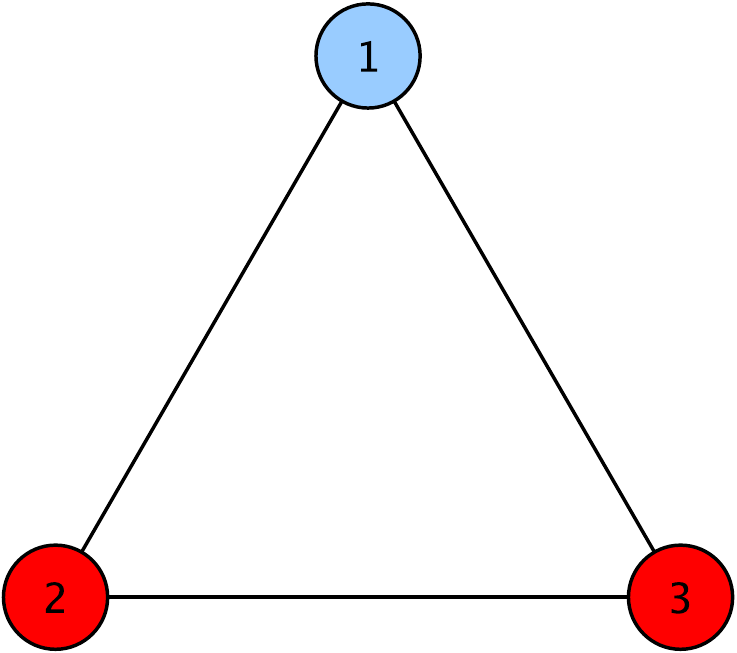}}
\quad\quad
\subfloat[]{\includegraphics[scale=0.45]{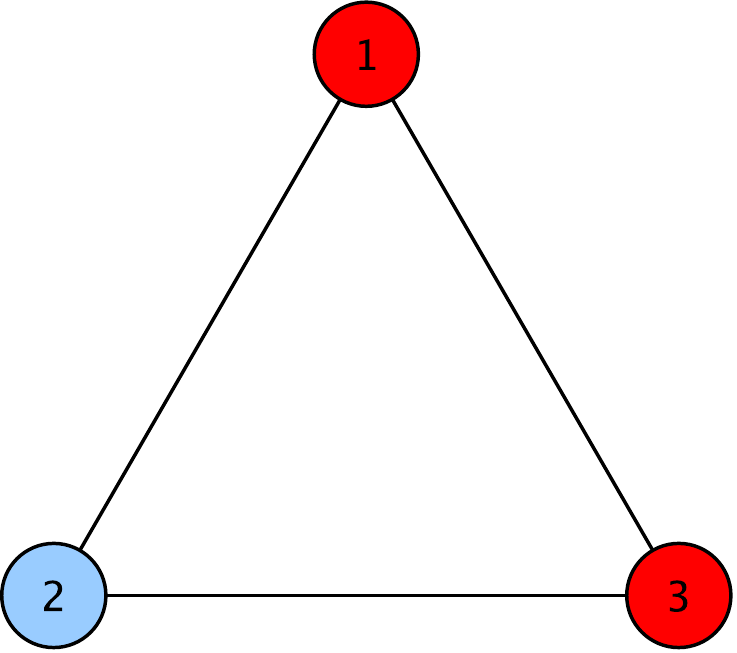}}
\end{center}
\caption{Two states consisting of a single cooperator (mutant) among defectors (the wild type) in a well-mixed population of size $N=3$. Despite the spatial symmetry of this population, an evolutionary game on this graph can be heterogeneous as a result of heterogeneous strategy mutations or asymmetric payoffs.\label{fig:wellMixedThree}}
\end{figure}

\subsubsection{Homogeneous evolutionary games}\label{subsubsec:homogeneousEvolutionaryGames}

The behavior of an evolutionary process sometimes depends heavily on the choice of update rule. As a result, a particular problem in evolutionary game theory is often stated (such as the evolution of cooperation) and subsequently explored separately for a number of different update rules. For example, consider the Donation Game (an instance of the Prisoner's Dilemma) in which cooperators pay a cost, $c$, in order to provide the opponent with a benefit, $b$. Defectors pay no costs and provide no benefits. On a large regular graph of degree $k$, \citet{ohtsuki:Nature:2006} show that selection favors cooperation in the death-birth process if $b/c>k$, but selection \textit{never} favors cooperation in the birth-death process. Therefore, the approach of exploring a problem in evolutionary game theory separately for several update rules has its merits. On the other hand, one might expect that high degrees of symmetry in the population structure, payoffs, and strategy mutations induce symmetries in an evolutionary game for a variety of update rules.

Before stating our main theorem for symmetric matrix games, we must first understand the basic components that make up an evolutionary game. Evolutionary games generally have two timescales: interactions and updates. In each (discrete) time step, every player in the population has a strategy, and this strategy profile determines the \textit{state} of the population. Neighbors then interact (quickly) and receive payoffs based on these strategies and the game(s) being played. The total payoff to a player determines his or her fitness. In the update step of the process, the strategies of the players are updated stochastically based on the fitness profile of the population, the population structure, and the strategy mutations. Popular examples of evolutionary update rules are birth-death, death-birth, imitation, pairwise comparison, and Wright-Fisher. Since interactions happen much more quickly than updates, there is a separation of timescales.

The most difficult part of an evolutionary game to describe in generality is the update step. If $S$ is the strategy set of the game and $N$ is the population size, then a state of the population is simply an element of $S^{N}$, i.e. a specification of a strategy for each member of the population. Implicit in the state space of the population being $S^{N}$ is an enumeration of the players. That is, if $\mathfrak{s}\in S^{N}$ is an $N$-tuple of strategies, then this profile indicates that player $i$ uses strategy $\mathfrak{s}_{i}$. For our purposes, we need only one property to be satisfied by the update rule of the process, which we state here as an axiom of an evolutionary game:
\begin{axiom}
The update rule of an evolutionary game is independent of the enumeration of the players.
\end{axiom}
\begin{remark}\label{axiomRemark}
As an example of what this axiom means, consider a death-birth process in which a player is selected uniformly at random for death and is replaced by the offspring of a neighbor. A neighbor is chosen for reproduction with probability proportional to fitness, and the offspring of this neighbor inherits the strategy of the parent with probability $1-\varepsilon$ and takes on a novel strategy with probability $\varepsilon$ for some $\varepsilon >0$. If all else is held constant (fitness, mutations, etc.), the fact that a player is referred to as the player at location $i$ is irrelevant: Let $\mathfrak{S}_{N}$ be the symmetric group on $N$ letters. If $\pi\in\mathfrak{S}_{N}$ is a permutation that relabels the locations of the players by sending $i$ to $\pi^{-1}\left(i\right)$, then the strategy of the player at location $\pi^{-1}\left(i\right)$ after the relabeling is the same as the strategy of the player at location $i$ before the relabeling. In particular, if $\mathfrak{s}\in S^{N}$ is the state of the population before the relabeling, then $\pi\left(\mathfrak{s}\right)\in S^{N}$ is the state of the population after the relabeling, where $\pi\left(\mathfrak{s}\right)_{i}=\mathfrak{s}_{\pi\left(i\right)}$. The probability that player $\pi^{-1}\left(i\right)$ is selected for death and replaced by the offspring of player $\pi^{-1}\left(j\right)$ after the relabeling is the same as the probability that player $i$ is selected for death and replaced by the offspring of player $j$ before the relabeling. Thus, for this death-birth process, the probability of transitioning between states $\mathfrak{s}$ and $\mathfrak{s}'$ before the relabeling is the same as the probability of transitioning between states $\pi\left(\mathfrak{s}\right)$ and $\pi\left(\mathfrak{s}'\right)$ after the relabeling. In this sense, a relabeling of the players induces an automorphism of the Markov chain defined by the process (in the sense of Definition \ref{def:symmetry}), and the axiom states that this phenomenon should hold for \textit{any} evolutionary update rule.
\end{remark}

In order to state our main result for symmetric games, we note that an evolutionary graph, $\Gamma$, in this setting consists of \textit{two} graphs: $\mathscr{E}$ and $\mathscr{D}$. We say that $\Gamma$ is regular if both $\mathscr{E}$ and $\mathscr{D}$ are regular. For vertex-transitivity of $\Gamma$ (resp. symmetry of $\Gamma$), we require slightly more than each $\mathscr{E}$ and $\mathscr{D}$ being vertex-transitive (resp. symmetric); we require them to be \textit{simultaneously} vertex-transitive (resp. symmetric). First of all, we need to define what an \textit{automorphism} of $\Gamma$ is. For $\pi\in\mathfrak{S}_{N}$, let $\pi\mathscr{E}$ be the graph defined by $\left(\pi\mathscr{E}\right)_{ij}:=\mathscr{E}_{\pi\left(i\right)\pi\left(j\right)}$ for each $i$ and $j$. Using this action, we define an automorphism of an evolutionary graph as follows:
\begin{definition}[Automorphism of an evolutionary graph]
An \textit{automorphism} of $\Gamma =\left(\mathscr{E},\mathscr{D}\right)$ is an action, $\pi\in\mathfrak{S}_{N}$, such that $\pi\mathscr{E}=\mathscr{E}$ and $\pi\mathscr{D}=\mathscr{D}$. We denote by $\textrm{Aut}\left(\Gamma\right)$ the set of automorphisms of $\Gamma$.
\end{definition}

We now have the definitions of vertex-transitive and symmetric evolutionary graphs:
\begin{definition}
$\Gamma =\left(\mathscr{E},\mathscr{D}\right)$ is \textit{vertex-transitive} if for each $i$ and $j$, there exists $\pi\in\textrm{Aut}\left(\Gamma\right)$ such that $\pi\left(i\right) =j$.
\end{definition}
\begin{definition}
$\Gamma =\left(\mathscr{E},\mathscr{D}\right)$ is \textit{symmetric} if $\mathscr{E}=\mathscr{D}$ and $\mathscr{E}$ is a symmetric graph.
\end{definition}

Finally, using the notion of an automorphism of $\Gamma$, we have our main result:
\begin{theorem}\label{thm:mainSymmetricTheorem}
Consider an evolutionary matrix game on a graph, $\Gamma =\left(\mathscr{E},\mathscr{D}\right)$, with symmetric payoffs and homogeneous strategy mutations. If $\pi\in\textrm{Aut}\left(\Gamma\right)$, then the states with a single mutant at vertex $i$ and $\pi\left(i\right)$, respectively, in an otherwise-monomorphic population, are evolutionarily equivalent. That is, in the notation of Definition \ref{def:homogeneousEvolutionaryProcess}, the states $\mathfrak{s}_{\left(s',i\right) ,s}$ and $\mathfrak{s}_{\left(s',\pi\left(i\right)\right) ,s}$ are evolutionarily equivalent for each $s,s'\in S$.
\end{theorem}
The proof of Theorem \ref{thm:mainSymmetricTheorem} may be found in Appendix B. The proof relies on the observation that the hypotheses of the theorem imply that any two states consisting of a single $A$-player in a population of $B$-players can be obtained from one another by relabeling the players. Thus, in light of the argument in Remark \ref{axiomRemark}, relabeling the players induces an automorphism on the Markov chain defined by the evolutionary game. Since any relabeling of the players leaves the monomorphic states fixed, there is an evolutionary equivalence between any two such states in the sense of Definition \ref{def:evolutionaryEquivalence}. Note that this theorem makes no restrictions on the selection strength or the update rule.
\begin{corollary}\label{cor:mainSymmetricCorollary}
An evolutionary game on a vertex-transitive graph with symmetric payoffs and homogeneous strategy mutations is itself homogeneous.
\end{corollary}

\begin{remark}
By Theorem \ref{thm:mainSymmetricTheorem}, two mutants appearing on a graph might define evolutionarily equivalent states even if the graph is not vertex-transitive. For example, the Tietze graph \citep[see][]{bondy:S:2008}, like the Frucht graph, has $12$ vertices and is regular but not vertex-transitive. However, unlike the Frucht graph, the Tietze graph has some nontrivial automorphisms. By Theorem \ref{thm:mainSymmetricTheorem}, any two vertices in the Tietze graph that are related by an automorphism have the property that the two corresponding single-mutant states are indistinguishable. An example of two evolutionarily equivalent states on this graph is given in Fig. \ref{fig:tietze}. In Appendix C, for the Snowdrift Game with death-birth updating, we give the fixation probabilities and absorption times for all configurations of a single cooperator among defectors, which further illustrates the effects of graph symmetries on an evolutionary process.
\end{remark}

\begin{figure}
\begin{center}
\subfloat[]{\includegraphics[scale=0.4]{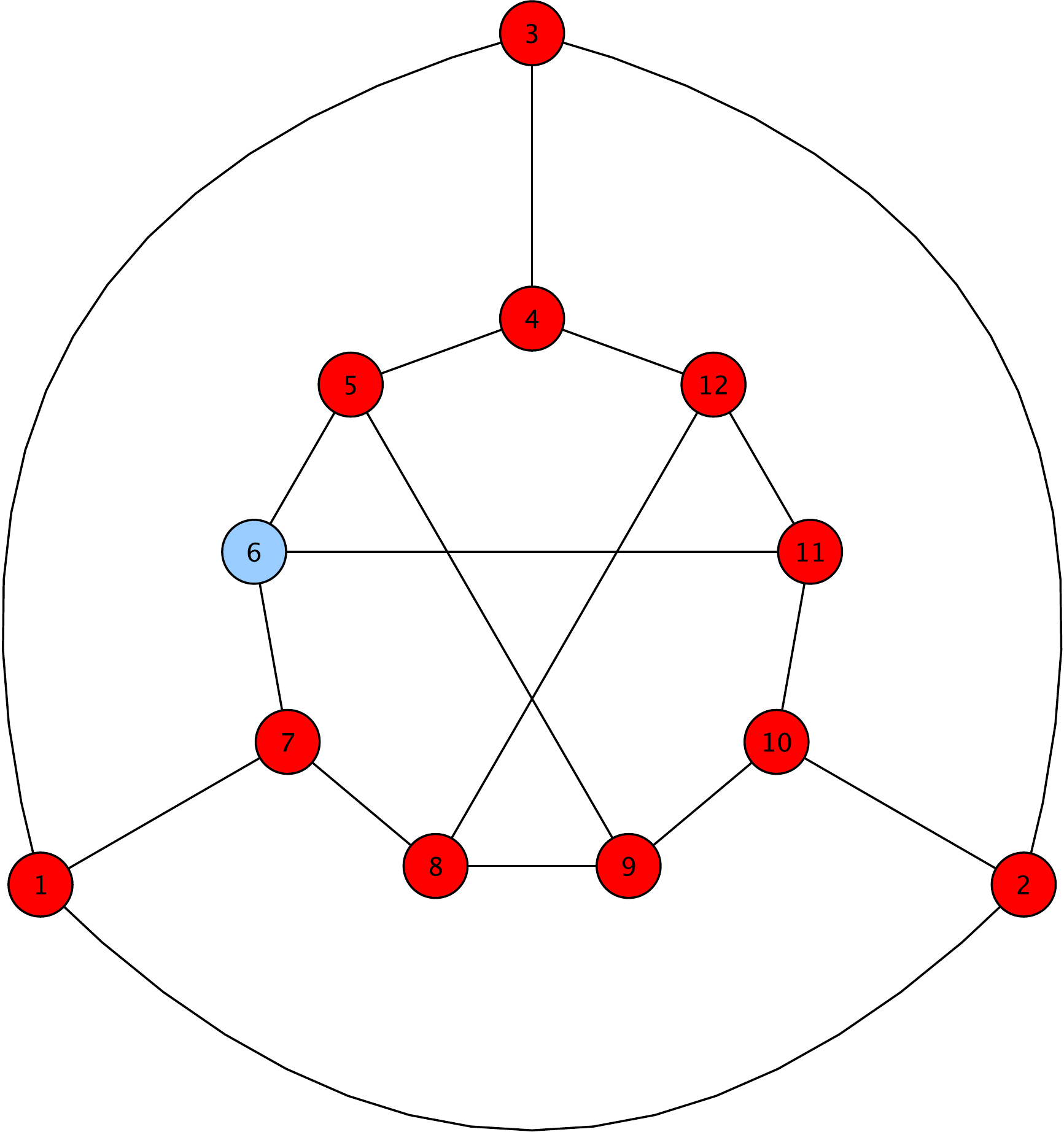}}
\quad\quad
\subfloat[]{\includegraphics[scale=0.4]{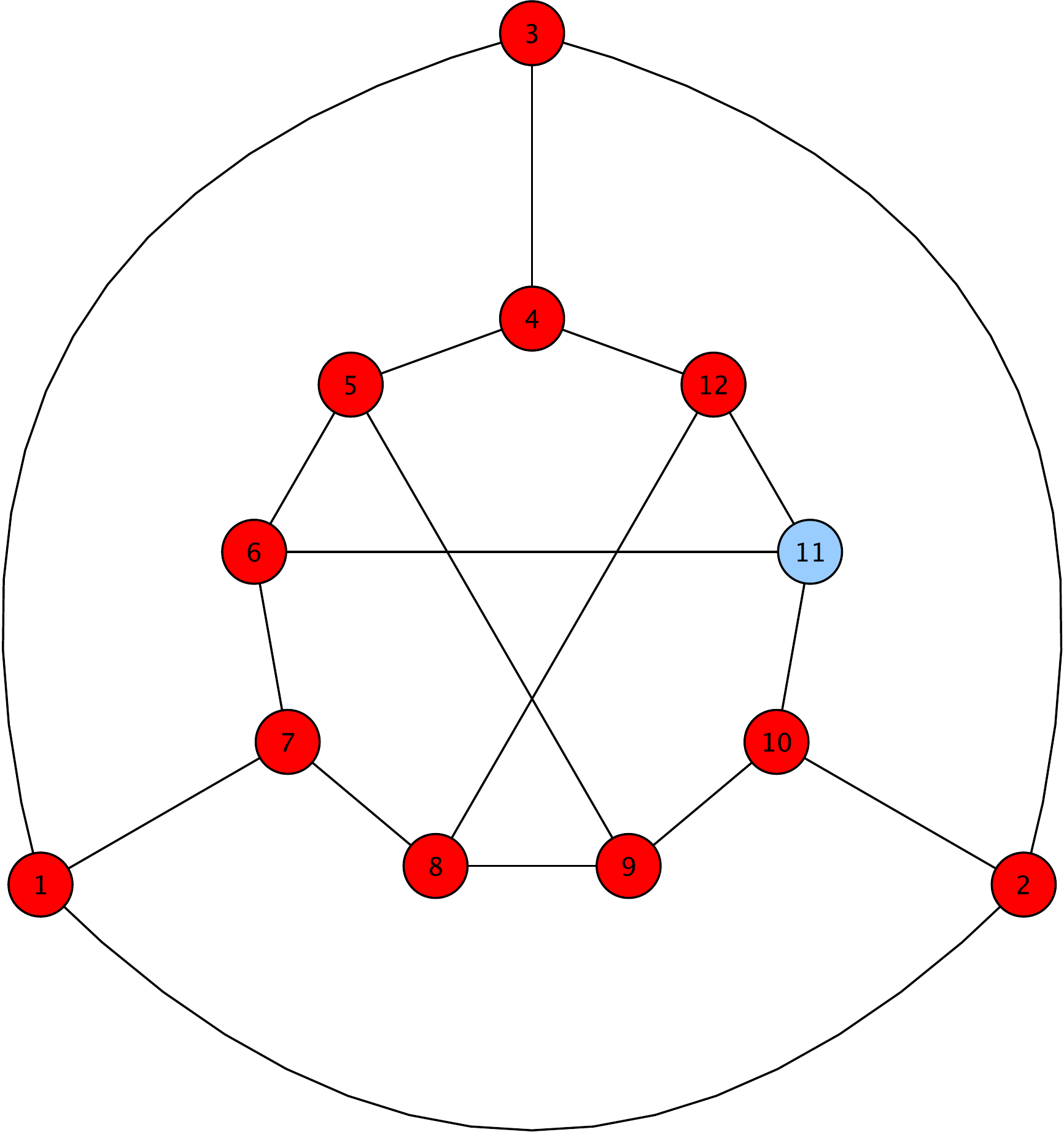}}
\end{center}
\caption{The Tietze graph with two different initial configurations. Like the Frucht graph, the Tietze graph is regular of degree $k=3$ (but not vertex-transitive) with $N=12$ vertices. Unlike the Frucht graph, the Tietze graph possesses nontrivial automorphisms. In (A), a cooperator is at vertex $6$ and all other players are defectors. In (B), a cooperator is at vertex $11$ and, again, the other players are defectors. Despite the fact that the Tietze graph is not vertex-transitive, the single-mutant states defined by (A) and (B) are evolutionarily equivalent. Graphically, this result is clear since one obtains (A) from (B) by flipping the graph (i.e. applying an automorphism), and such a difference between the two states does not affect fixation probabilities, times, etc. However, it is not true that any two single-mutant states are evolutionarily equivalent. For example, in the Snowdrift Game with $\beta =0.1$ and death-birth updating, the single-mutant state with a cooperator at vertex $1$ (resp. vertex $6$) has a fixation probability of $0.3777$ (resp. $0.4186$). Therefore, the two single-mutant states with cooperators at vertices $1$ and $6$, respectively, are not evolutionarily equivalent, so this process is not homogeneous.\label{fig:tietze}}
\end{figure}

\begin{remark}\label{remark:transitiveSymmetric}
For a given population size, $N$, and network degree, $k$, there may be many vertex-transitive graphs of size $N$ with degree $k$. For each such graph, the fixation probability of a randomly occuring mutant is independent of where on the graph it occurs by Theorem \ref{thm:mainSymmetricTheorem}. However, this fixation probability depends on more than just $N$ and $k$; it also depends on the configuration of the network. For example, Fig. \ref{fig:transitiveSymmetric} gives two vertex-transitive graphs of size $N=6$ and degree $k=3$. As an illustration, consider the Snowdrift Game (with payoff matrix (\ref{snowdriftGame})) on these graphs with birth-death updating. If the selection intensity is $\beta =0.1$, then the fixation probability of a single cooperator in a population of defectors is $0.1632$ in (A) and $0.1722$ in (B) (both rounded to four significant figures). These two fixation probabilities differ for all but finitely many $\beta\geqslant 0$ by Proposition \ref{prop:functionOfBeta}.
\end{remark}

\begin{figure}
\begin{center}
\subfloat[]{\includegraphics[scale=0.5]{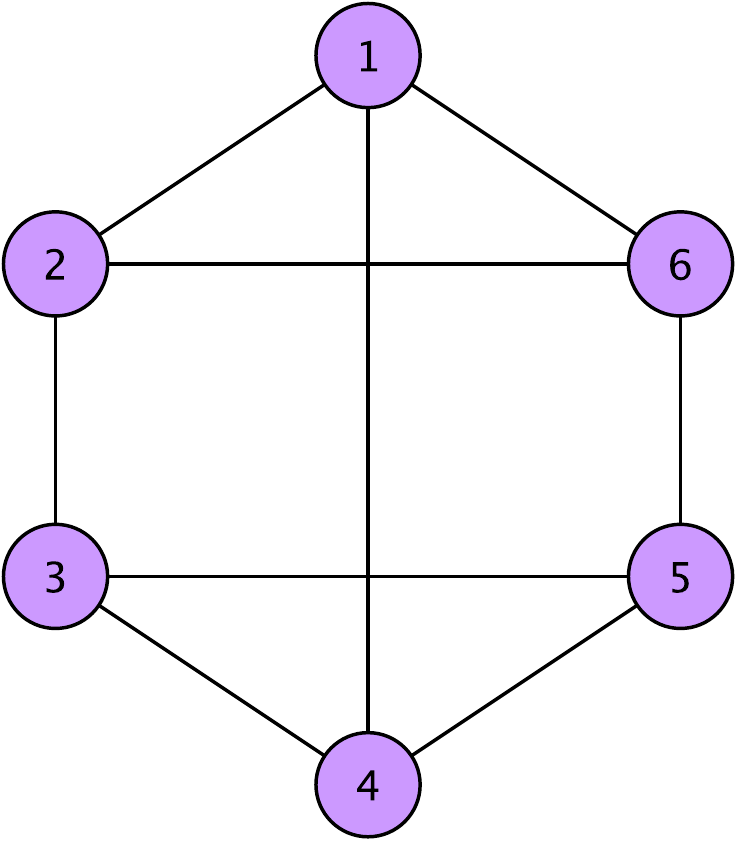}}
\quad\quad
\subfloat[]{\includegraphics[scale=0.5]{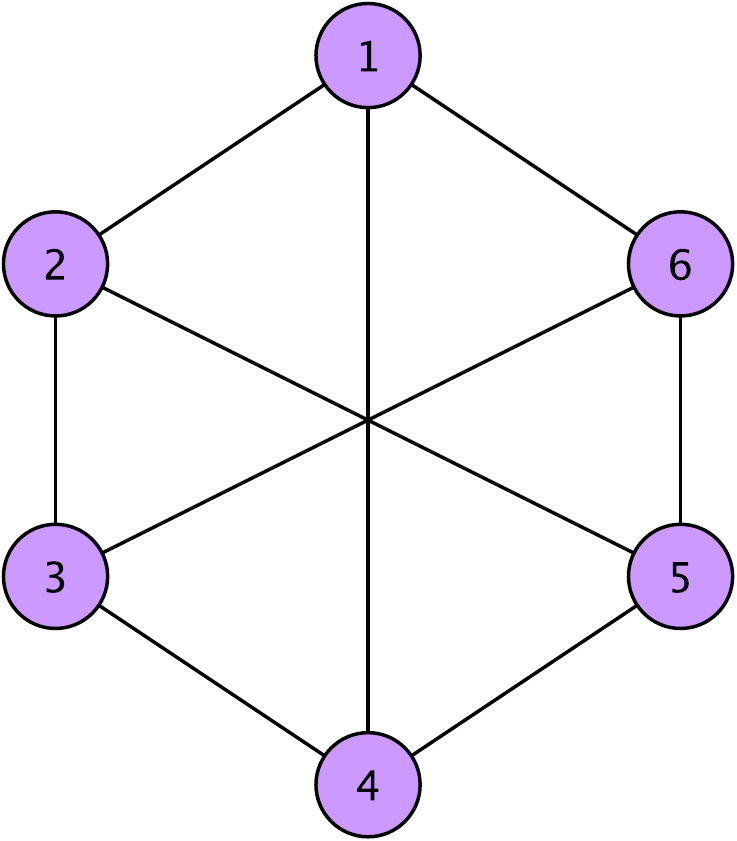}}
\end{center}
\caption{Undirected, unweighted, vertex-transitive graphs of degree $k=3$ with $N=6$ vertices. (B) is a symmetric (arc-transitive) graph and (A) is not.\label{fig:transitiveSymmetric}}
\end{figure}

Until this point, our focus has been on states consisting of just a single mutant in an otherwise-monomorphic population. One could also inquire as to when any two states consisting of two (or three, four, etc.) mutants are evolutionarily equivalent. It turns out that that the answer to this question is simple: in general, the population must be well-mixed in order for any two $m$-mutant states to be evolutionarily equivalent if $m>1$. The proof that this equivalence holds in well-mixed populations follows from the argument given to establish Theorem \ref{thm:mainSymmetricTheorem} (see Appendix B). On the other hand, if the population is not well-mixed, then one can find a pair of states with the first state consisting of two mutants on neighboring vertices and the second state consisting of two mutants on non-neighboring vertices. In general, the mutant type will have different fixation probabilities in these two states. For example, in the Snowdrift Game on the graph of Fig. \ref{fig:transitiveSymmetric}(B), consider the two states, $\mathfrak{s}$ and $\mathfrak{s}'$, where $\mathfrak{s}$ consists of cooperators on vertices $1$ and $2$ only and $\mathfrak{s}'$ consists of cooperators on vertices $1$ and $3$ only. If $\beta =0.1$, then the fixation probability of cooperators under death-birth updating when starting at $\mathfrak{s}$ (resp. $\mathfrak{s}'$) is $0.3126$ (resp. $0.2607$). Therefore, despite the arc-transitivity of this graph, it is not true that any two states consisting of exactly two mutants are evolutionarily equivalent. Only in well-mixed populations are we guaranteed that any two such states are equivalent.

\subsection{Asymmetric games}\label{subsec:asymmetricGames}

One particular form of payoff asymmetry appearing in evolutionary game theory is \textit{ecological asymmetry} \citep{mcavoy:PLOSCB:2015}. Ecological asymmetry can arise as a result of an uneven distribution of resources. For example, in the Donation Game, a cooperator at location $i$ might provide a benefit to his or her opponent based on some resource derived from the environment. Both this resource and the cost of donating it could depend on $i$, which means that different players have different payoff matrices. These payoff matrices depend on both the location of the focal player and the locations of the opponents. Thus, payoffs for a player at location $i$ against an opponent at location $j$ in an $n$-strategy ``bimatrix" game \citep{hofbauer:JMB:1996,ohtsuki:JTB:2010,mcavoy:PLOSCB:2015} are given by the asymmetric payoff matrix
\begin{linenomath}
\begin{align}
\mathbf{M}^{ij} = \bordermatrix{%
 & A_{1} & A_{2} & \cdots & A_{n} \cr
A_{1} &\ a_{11}^{ij}, a_{11}^{ji} & \ a_{12}^{ij}, a_{21}^{ji} & \ \cdots & \ a_{1n}^{ij}, a_{n1}^{ji} \cr
A_{2} &\ a_{21}^{ij}, a_{12}^{ji} & \ a_{22}^{ij}, a_{22}^{ji} & \ \cdots & \ a_{2n}^{ij}, a_{n2}^{ji} \cr
\ \vdots &\ \vdots & \ \vdots & \ \ddots & \ \vdots \cr
A_{n} &\ a_{n1}^{ij}, a_{1n}^{ji} & \ a_{n2}^{ij}, a_{2n}^{ji} & \ \cdots & \ a_{nn}^{ij}, a_{nn}^{ji} \cr
} .
\end{align}
\end{linenomath}
Similar to Eq. (\ref{eq:symmetricMatrixPayoff}), the total payoff to player $i$ for strategy profile $\left(s_{1},\dots ,s_{N}\right)\in\left\{1,\dots ,n\right\}^{N}$ is
\begin{linenomath}
\begin{align}\label{eq:asymmetricMatrixPayoff}
u_{i}\left(s_{1},\dots ,s_{N}\right) &:= \sum_{j=1}^{N}\mathscr{E}_{ij}a_{s_{i}s_{j}}^{ij} .
\end{align}
\end{linenomath}

We saw in \S\ref{subsec:frequencyDependentGames} an example of a heterogeneous evolutionary game in a well-mixed population with symmetric payoffs. Rather than looking at a symmetric game with heterogeneous strategy mutations, we now look at an asymmetric game with homogeneous strategy mutations. Consider the ecologically asymmetric Donation Game on the graph of Fig. \ref{fig:wellMixedThree} (both $\mathscr{E}$ and $\mathscr{D}$) with a death-birth update rule. In this asymmetric Donation Game, a cooperator at location $i$ donates $b_{i}$ at a cost of $c_{i}$; defectors donate nothing and incur no costs. If $\beta =0.1$, $b_{1}=b_{2}=b_{3}=4$, $c_{1}=1$, and $c_{2}=c_{3}=3$, then the fixation probability of a single cooperator at location $1$ (see Fig. \ref{fig:wellMixedThree}(A)) is $0.2232$, while the fixation probability of a single cooperator at location $2$ (see Fig. \ref{fig:wellMixedThree}(B)) is $0.1842$ (both rounded to four significant figures). Therefore, even in a well-mixed population with homogeneous strategy mutations (none, in this case), asymmetric payoffs can prevent an evolutionary game from being homogeneous.

Asymmetric matrix games in large populations reduce to symmetric games if selection is weak \citep{mcavoy:PLOSCB:2015}. In the limit of weak selection, \citet{mcavoy:JMB:2015} establish a selection condition for asymmetric matrix games in finite graph-structured populations that extends the condition (for symmetric games) of \citet{tarnita:PNAS:2011}:
\begin{theorem}[\citet{mcavoy:JMB:2015}]\label{asymmetricSelectionCondition}
There exists a set of \textit{structure coefficients}, $\left\{\sigma_{1}^{ij},\sigma_{2}^{ij},\sigma_{3}^{ij}\right\}_{i,j}$, independent of payoffs, such that weak selection favors strategy $r\in\left\{1,\dots ,n\right\}$ if and only if
\begin{linenomath}
\begin{align}\label{asymmetricNetwork}
\sum_{i,j=1}^{N}\mathscr{E}_{ij}\left( \sigma_{1}^{ij}\left( a_{rr}^{ij}-\overline{a}_{\ast\ast}^{ij} \right) + \sigma_{2}^{ij}\left(\overline{a}_{r\ast}^{ij}-\overline{a}_{\ast r}^{ij}\right) + \sigma_{3}^{ij}\left(\overline{a}_{r\ast}^{ij}-\overline{a}^{ij}\right) \right) &> 0 ,
\end{align}
\end{linenomath}
where $\overline{a}_{\ast\ast}^{ij}=\frac{1}{n}\sum_{s=1}^{n}a_{ss}^{ij}$, $\overline{a}_{r\ast}^{ij}=\frac{1}{n}\sum_{s=1}^{n}a_{rs}^{ij}$, $\overline{a}_{\ast r}^{ij}=\frac{1}{n}\sum_{s=1}^{n}a_{sr}^{ij}$, and $\overline{a}^{ij}=\frac{1}{n^{2}}\sum_{s,t=1}^{n}a_{st}^{ij}$.
\end{theorem}
Strictly speaking, Theorem \ref{asymmetricSelectionCondition} is established for $\mathscr{E}$ and $\mathscr{D}$ undirected, unweighted, and satisfying $\mathscr{E}=\mathscr{D}$. However, the proof of Theorem \ref{asymmetricSelectionCondition} extends immediately to the case with $\mathscr{E}$ and $\mathscr{D}$ directed, weighted, and possibly distinct, so we make no restrictive assumptions on $\mathscr{E}$ and $\mathscr{D}$ in the statement of this theorem here. In the simpler case $n=2$, condition (\ref{asymmetricNetwork}) takes the form
\begin{linenomath}
\begin{align}\label{asymmetricNetworkTwo}
\sum_{i,j=1}^{N}\mathscr{E}_{ij}\left( \tau_{1}^{ij}\left(a_{11}^{ij}-a_{22}^{ij}\right) + \tau_{2}^{ij}\left(a_{12}^{ij}-a_{21}^{ij}\right) \right) &> 0
\end{align}
\end{linenomath}
for some collection $\left\{\tau_{1}^{ij},\tau_{2}^{ij}\right\}_{i,j}$. For the death-birth process with $\mathscr{E}$ and $\mathscr{D}$ the graph of Fig. \ref{fig:transitiveSymmetric}(A), we calculate \textit{exact} values for all of these structure coefficients (see Appendix C). In particular, we find that $\tau_{1}^{12}=707905/9315552$ and $\tau_{1}^{14}=16291/194074$, so vertex-transitivity does not guarantee that the structure coefficients are independent of $i$ and $j$. For the same process on the graph in Fig. \ref{fig:transitiveSymmetric}(B), we find that $\tau_{1}^{ij}=\tau_{2}^{ij}=2189/27728$ for each $i$ and $j$, so these coefficients do not depend on $i$ and $j$. (In general, even for well-mixed populations, $\tau_{1}$ and $\tau_{2}$ need not be the same; for the same process studied here but on the graph of Fig. \ref{fig:wellMixedThree}, $\tau_{1}^{ij}=33/1616$ and $\tau_{2}^{ij}=99/1616$ for each $i$ and $j$.) This lack of dependence on $i$ and $j$ is due to the fact that the graph of Fig. \ref{fig:transitiveSymmetric}(B) is \textit{symmetric}, and it turns out to be a special case of a more general result:
\begin{theorem}\label{thm:mainAsymmetricTheorem}
Suppose that an asymmetric matrix game with homogeneous strategy mutations is played on an evolutionary graph, $\Gamma =\left(\mathscr{E},\mathscr{D}\right)$. For each $\pi\in\textrm{Aut}\left(\Gamma\right)$, $k\in\left\{1,2,3\right\}$, and $i,j\in\left\{1,\dots ,N\right\}$,
\begin{linenomath}
\begin{align}
\sigma_{k}^{ij} &= \sigma_{k}^{\pi\left(i\right)\pi\left(j\right)} .
\end{align}
\end{linenomath}
\end{theorem}
The proof of Theorem \ref{thm:mainAsymmetricTheorem} may be found in Appendix B. The following corollary is an immediate consequence of Theorem \ref{thm:mainAsymmetricTheorem}:
\begin{corollary}
If $\mathscr{E}=\mathscr{D}$ and $\mathscr{E}$ is a symmetric graph (i.e. $\Gamma$ is a symmetric evolutionary graph), then the structure coefficients are independent of $i$ and $j$.
\end{corollary}
Since symmetric graphs are also regular, we have:
\begin{corollary}\label{cor:spatialAverage}
If $\mathscr{E}=\mathscr{D}$ and $\mathscr{E}$ is a symmetric graph (i.e. $\Gamma$ is a symmetric evolutionary graph), then strategy $r$ is favored in the limit of weak selection if and only if
\begin{align}\label{tarnitaSelectionCondition}
\sigma_{1}\left( a_{rr}-\overline{a}_{\ast\ast} \right) + \sigma_{2}\left(\overline{a}_{r\ast}-\overline{a}_{\ast r}\right) + \sigma_{3}\left(\overline{a}_{r\ast}-\overline{a}\right) &> 0 ,
\end{align}
where $\overline{\mathbf{M}}=\left(a_{st}\right)_{1\leqslant s,t\leqslant n}$ is the spatial average of the matrices $\mathbf{M}^{ij}$, i.e.
\begin{linenomath}
\begin{align}
\overline{\mathbf{M}} &:= \frac{1}{kN}\sum_{i,j=1}^{N}\mathscr{E}_{ij}\mathbf{M}^{ij} ,
\end{align}
\end{linenomath}
where $k$ is the degree of the graph, $\Gamma$.
\end{corollary}

\begin{remark}
Eq. (\ref{tarnitaSelectionCondition}) is just the selection condition of \citet{tarnita:PNAS:2011} for symmetric matrix games.
\end{remark}

It follows from Corollary \ref{cor:spatialAverage} that asymmetric matrix games on arc-transitive (symmetric) graphs can be reduced to symmetric games in the limit of weak selection.

\section{Discussion}

Evolutionary games in finite populations may be split into two classes: those with absorbing states (``absorbing processes") and those without absorbing states. In absorbing processes, the notion of fixation probability has played a crucial role in quantifying evolutionary outcomes, but fixation probabilities are far from the only measure of evolutionary success. Much of the literature on evolutionary games with absorbing states has neglected other metrics such as the time to absorption or the time to fixation conditioned on fixation occurring (``conditional fixation time"). This bias toward fixation probabilities has resulted in certain evolutionary processes appearing more symmetric than they actually are. We have illustrated this phenomenon using the frequency-independent Moran process on graphs: The Isothermal Theorem guarantees that, on regular graphs, a single mutant cannot distinguish between initial locations in the graph if the only metric under consideration is the probability of fixation. However, certain initial placements of the mutant may result in faster absorption times than others if the graph is regular but not vertex-transitive, and the Frucht graph exemplifies this claim. The same phenomenon also holds for conditional fixation times.

The Frucht graph, which is a regular structure with no nontrivial symmetries (there are no two vertices from which the graph ``looks" the same), also allowed us to show that the Isothermal Theorem of \citet{lieberman:Nature:2005} does not extend to frequency-dependent evolutionary games. That is, on regular graphs, the probability of fixation of a single mutant may depend on the initial location of the mutant if fitness is frequency-dependent. This claim was illustrated via a death-birth process on the Frucht graph, in which the underlying evolutionary game was a Snowdrift Game. For $\beta =1$ (strong selection), the fixation probability of a cooperator at vertex $11$ was nearly $14\%$ larger than the fixation probability of a cooperator at vertex $4$. Moreover, we showed that if the fixation probabilities of two initial configurations differ for a single value of $\beta$, then they are the same for at most finitely many values of $\beta$. In particular, these fixation probabilities differ for almost every selection strength, so our observation for $\beta =1$ was not an anomaly. Similar phenomena are observed for frequency-dependent birth-death processes on the Frucht graph, for example, and even for frequency-dependent games with the ``equal gains from switching" property, such as the Donation Game.

Theorem \ref{thm:mainSymmetricTheorem} is an analogue of the Isothermal Theorem that applies to a broader class of games and update rules. The Isothermal Theorem is remarkable since regularity of the population structure implies that the fixation probabilities are not only independent of the initial location of the mutant, they are the \textit{same} as those of the classical Moran process. Our treatment of homogeneous evolutionary processes is focused on when different single-mutant states are equivalent, not when they are equivalent to the corresponding states in the classical Moran process. Even if the fixation probability of a single mutant does not depend on the mutant's location, other factors (such as birth and death rates) may affect whether or not this fixation probability is the same as the one in a well-mixed population \citep{komarova:BMB:2006,kaveh:RSOS:2015}. Remark \ref{remark:transitiveSymmetric}, which compares the fixation probabilities for the Snowdrift Game on two different vertex-transitive graphs of the same size and degree, shows that the fixation probability of a single mutant--even if independent of the mutant's location--can depend on the configuration of the graph. In light of these results, the symmetry phenomena for the Moran process guaranteed by the Isothermal Theorem do not generalize and should be thought of as properties of the frequency-independent Moran process and not of evolutionary processes in general.

Theorem \ref{thm:mainSymmetricTheorem}, and indeed most of our discussion of homogeneity, focused on symmetries of states consisting of just a \textit{single} mutant. In many cases, mutation rates are sufficiently small that a mutant type, when it appears, will either fixate or go extinct before another mutation occurs \citep{fudenberg:JET:2006,wu:JMB:2011}. Thus, with small mutation rates, one need not consider symmetries of states consisting of more than one mutant. However, if mutation rates are larger, then these multi-mutant states become relevant. Our definition of evolutionary equivalence (Definition \ref{def:evolutionaryEquivalence}) applies to these states, but, as expected, the symmetry conditions on the population structure guaranteeing any two multi-mutant states are equivalent are much stronger. In fact, as we argued in \S\ref{subsubsec:homogeneousEvolutionaryGames}, the population must in general be well-mixed even for any pair of states consisting of \textit{two} mutants to be evolutionarily equivalent. Consequently, our focus on single-mutant states allowed us to simultaneously treat biologically relevant configurations (assuming mutation rates are small) and obtain non-trivial conditions guaranteeing homogeneity of an evolutionary process.

The counterexamples presented here could be defined on sufficiently small population structures, and thus all calculations (fixation probabilities, structure coefficients, etc.) are exact. However, these quantities need not always be explicitly calculated in order to prove useful: In our study of asymmetric games, we concluded that an asymmetric game on an arc-transitive (symmetric) graph can be reduced to a symmetric game in the limit of weak selection. (The graph of Fig. \ref{fig:transitiveSymmetric}(A) demonstrates that vertex-transitivity alone does not guarantee that an asymmetric game can be reduced to a symmetric game in this way.) This result was obtained by examining the qualitative nature of the structure coefficients in the selection condition (\ref{asymmetricNetwork}), but it did not require explicit calculations of these coefficients. Therefore, despite the difficulty in actually calculating these coefficients, they can still be used to glean qualitative insight into the dynamics of evolutionary games.

On large random regular graphs, the dynamics of an asymmetric matrix game are equivalent to those of a certain symmetric game obtained as a ``spatial average" of the individual asymmetric games \citep{mcavoy:PLOSCB:2015}. Corollary \ref{cor:spatialAverage} is highly reminiscent of this type of reduction to a symmetric game. For large populations, this result is obtained by observing that large random regular graphs approximate a Bethe lattice \citep{bollobas:CUP:2001} and then using the pair approximation method \citep{matsuda:PTP:1992} to describe the dynamics. The pair approximation method is \textit{exact} for a Bethe lattice \citep{ohtsuki:Nature:2006}, so, from this perspective, Corollary \ref{cor:spatialAverage} is not that surprising since Bethe lattices are arc-transitive. Of course, a Bethe lattice has infinitely many vertices, and Corollary \ref{cor:spatialAverage} is a finite-population analogue of this result.

The term ``homogeneous" is used in the literature to refer to several different kinds of population structures. This term has been used to describe well-mixed populations \citep{assaf:PRL:2012,szolnoki:PRE:2014b}. For graph-structured populations, ``homogeneous graph" sometimes refers to vertex-transitive graphs \citep{taylor:Nature:2007,tarnita:TAN:2014}. In algebraic graph theory, however, the term ``homogeneous graph" implies a much higher degree of symmetry than does than vertex-transitivity \citep[see][]{beineke:CUP:2004}. ``Homogeneous" has also been used to describe graphs in which each vertex has the same number of neighbors, i.e. regular graphs \citep{roca:PLR:2009,hindersin:JRSI:2014,cheng:IEEE:2015}. In between regular and vertex-transitive graphs, ``homogeneous graph" has also referred to large, random regular graphs \citep{traulsen:AN:2009}. As we noted, large, random regular graphs approximate Bethe lattices (which are infinite, arc-transitive graphs), but these approximations need not themselves be even vertex-transitive.

In many of the various uses of the term ``homogeneous," a common aim is to study the fixation probability of a randomly placed mutant. Our definition of \textit{homogeneous evolutionary game} formally captures what it means for two single-mutant states to be equivalent, and our explorations of the Frucht graph (in conjunction with Theorem \ref{thm:mainSymmetricTheorem}) show that vertex-transitivity, and \textit{not} regularity, is what the term ``homogeneous" in graph-structured populations should indicate. We also demonstrated the effects of payoffs and strategy mutations on the behavior of these single-mutant states and concluded that the term homogeneous should apply to the entire process rather than to just the population structure. The homogeneity (Theorem \ref{thm:mainSymmetricTheorem}) and symmetry (Theorem \ref{thm:mainAsymmetricTheorem}) results given here do not depend on the update rule, in contrast with results such as the symmetry of conditional fixation times in the Moran process of \citet{taylor:JTB:2006} or the Isothermal Theorem of \citet{lieberman:Nature:2005}. We now know that games on regular graphs are \textit{not} homogeneous, and we know precisely under which conditions the ``fixation probability of a randomly placed mutant" is well-defined. These results provide a firmer foundation for evolutionary game theory in finite populations and a basis for defining the evolutionary success of the strategies of a game.

\section*{Acknowledgements}

The authors thank Wes Maciejewski for helpful conversations and the anonymous referees for their comments.

\section*{Funding statement}

A. M. and C. H. acknowledge financial support from the Natural Sciences and Engineering Research Council of Canada (NSERC), grant RGPIN-2015-05795, and C.H. from the Foundational Questions in Evolutionary Biology Fund (FQEB), grant RFP-12-10.

\bibliographystyle{plainnat}

\setcounter{section}{0}
\setcounter{subsection}{0}
\setcounter{equation}{0}
\renewcommand{\thesubsection}{A.\arabic{subsection}}
\renewcommand{\theequation}{A.\arabic{equation}}

\section*{Appendix A: fixation and absorption}

Using a method inspired by a technique of \citet{press:PNAS:2012}, we derive explicit expressions (in terms of the transition matrix) for fixation probabilities and absorption times. Subsequently, we prove a simple lemma that says that Markov chain symmetries preserve the set of a chain's stationary distributions.

\subsection{Fixation probabilities} Suppose that $\left\{X_{n}\right\}_{n\geqslant 0}$ is a discrete-time Markov chain on a finite state space, $\mathcal{S}$, that has exactly $K$ ($\geqslant 1$) absorbing states, $\mathfrak{s}_{1},\dots ,\mathfrak{s}_{K}$. Moreover, suppose that the non-absorbing states are transient \citep[see][]{fudenberg:JET:2006}. The transition matrix for this chain, $\mathbf{T}$, may be written as
\begin{linenomath}
\begin{align}
\mathbf{T} &= \begin{pmatrix}\mathbf{I}_{K} & \mathbf{0} \\ \mathbf{S}_{1} & \mathbf{S}_{2}\end{pmatrix} ,
\end{align}
\end{linenomath}
where $\mathbf{I}_{K}$ is the $K\times K$ identity matrix and $\mathbf{0}$ is the matrix of zeros (in this case, its dimension is $K\times\left(\left|\mathcal{S}\right| -K\right)$, where $\left|\mathcal{S}\right|$ is the number of states in $\mathcal{S}$). This chain will eventually end up in one of the $K$ absorbing states, and we denote by $\rho_{\mathfrak{s},i}$ the probability that state $\mathfrak{s}_{i}$ is reached when the chain starts off in states $\mathfrak{s}\in\mathcal{S}$. Let $\mathbf{P}$ be the $\left|\mathcal{S}\right|\times K$ matrix of fixation probabilities, i.e. $\mathbf{P}_{\mathfrak{s},i}=\rho_{\mathfrak{s},i}$ for each $\mathfrak{s}$ and $i$. This matrix satisfies
\begin{linenomath}
\begin{align}
\mathbf{T}\mathbf{P} &= \mathbf{P} ,
\end{align}
\end{linenomath}
which is just the matrix form of the recurrence relation satisfied by fixation probabilities (obtained from a ``first step" analysis of the Markov chain). Consider the matrix, $\mathcal{M}=\mathcal{M}\left(\mathbf{T}\right)$, defined by
\begin{linenomath}
\begin{align}
\mathcal{M} &:= -\begin{pmatrix}\mathbf{I}_{K} & \mathbf{0} \\ \mathbf{S}_{1} & \mathbf{S}_{2}-\mathbf{I}_{\left|\mathcal{S}\right| -K}\end{pmatrix} .
\end{align}
\end{linenomath}
Since $\rho_{j,i}=\delta_{i,j}$ for $i,j\in\left\{1,\dots ,K\right\}$, we see that
\begin{linenomath}
\begin{align}
-\mathcal{M}\mathbf{P} &= \mathbf{T}\mathbf{P} - \begin{pmatrix}\mathbf{0} & \mathbf{0} \\ \mathbf{0} & \mathbf{I}_{\left|\mathcal{S}\right| -K}\end{pmatrix}\mathbf{P} = \begin{pmatrix}\mathbf{I}_{K} \\ \mathbf{0}\end{pmatrix} .
\end{align}
\end{linenomath}
Moreover, the matrix $\mathcal{M}$ must have full rank since the non-absorbing states are transient; that is,
\begin{linenomath}
\begin{align}
\det\mathcal{M} &= \left(-1\right)^{\left|\mathcal{S}\right|}\det\left(\mathbf{S}_{2}-\mathbf{I}_{\left|\mathcal{S}\right| -K}\right) \neq 0 .
\end{align}
\end{linenomath}
Therefore, by Cramer's rule,
\begin{linenomath}
\begin{align}\label{fixationProbabilityFormula}
\rho_{\mathfrak{s},i} &= -\frac{\det\mathcal{M}\left(\mathfrak{s},\mathbf{e}_{i}\right)}{\det\mathcal{M}} ,
\end{align}
\end{linenomath}
where the notation $\mathcal{M}\left(\mathfrak{s},\mathbf{e}_{i}\right)$ means the matrix obtained by replacing the column corresponding to state $\mathfrak{s}$ with the $i$th standard basis vector, $\mathbf{e}_{i}$. Thus, Eq. (\ref{fixationProbabilityFormula}) gives explicit formulas for the fixation probabilities.

\subsection{Absorption times}

Let $\mathbf{t}$ be the $\left|\mathcal{S}\right|$-vector indexed by $\mathcal{S}$ whose entry $\mathbf{t}_{\mathfrak{s}}$ is the expected time until the process fixates in one of the absorbing states when started in state $\mathfrak{s}\in\mathcal{S}$. This vector satisfies $\mathbf{t}_{i}=0$ for $i=1,\dots ,K$ as well as the recurrence relation $\mathbf{T}\left(\mathbf{t}+\mathbf{1}\right) = \mathbf{t} + \sum_{i=1}^{K}\mathbf{e}_{i}$. Therefore,
\begin{linenomath}
\begin{align}
\mathcal{M}\mathbf{t} &= \sum_{j=K+1}^{\left|\mathcal{S}\right|}\mathbf{e}_{j} ,
\end{align}
\end{linenomath}
so, by Cramer's rule,
\begin{linenomath}
\begin{align}\label{absorptionTimeFormula}
\mathbf{t}_{\mathfrak{s}} &= \sum_{j=K+1}^{\left|\mathcal{S}\right|}\frac{\det\mathcal{M}\left(\mathfrak{s},\mathbf{e}_{j}\right)}{\det\mathcal{M}} .
\end{align}
\end{linenomath}

We now turn to Markov chains defined by evolutionary games. Before proving Proposition \ref{prop:functionOfBeta}, we make two assumptions:
\begin{enumerate}

\item[(i)] The payoff-to-fitness mapping is of the form $f_{\beta}\left(\pi\right) =\exp\left\{\beta\pi\right\}$, $f_{\beta}\left(\pi\right) =\beta\pi$, $f_{\beta}\left(\pi\right) =1+\beta\pi$, or $f_{\beta}\left(\pi\right) =1-\beta +\beta\pi$, where $\beta$ denotes the intensity of selection and $\pi$ denotes payoff. (Of course, fitness can be defined in one of the latter three ways only if the payoffs are such that $f_{\beta}\left(\pi\right)\geqslant 0$.)

\item[(ii)] The update probabilities are rational functions of the fitness profile of the population.

\end{enumerate}

\begin{remark}
Assumptions (i) and (ii) are not at all restrictive in evolutionary game theory. Any process in which selection occurs with probability proportional to fitness will satisfy this rationality condition, and indeed all of the standard evolutionary processes (birth-death, death-birth, imitation, pairwise comparison, Wright-Fisher, etc.) have this property. The four payoff-to-fitness mappings are standard as well.
\end{remark}

With assumptions (i) and (ii) in mind, we have:
\begin{repeatedProposition}
Each of the equalities
\begin{linenomath}
\begin{subequations}
\begin{align}
\rho_{\mathfrak{s},i} &= \rho_{\mathfrak{s}',j} ; \\
\mathbf{t}_{\mathfrak{s}} &= \mathbf{t}_{\mathfrak{s}'}
\end{align}
\end{subequations}
\end{linenomath}
holds for either (i) \textit{every} $\beta\geqslant 0$ or (ii) at most finitely many $\beta\geqslant 0$. Thus, if one of these equalities fails to hold for even a single value of $\beta$, then it fails to hold for all sufficiently small $\beta >0$.
\end{repeatedProposition}
\begin{proof}
Suppose that $\mathfrak{s},\mathfrak{s}'\in\mathcal{S}$ and that $\mathfrak{s}_{i}$ and $\mathfrak{s}_{j}$ are absorbing states. By Eq. (\ref{fixationProbabilityFormula}),
\begin{linenomath}
\begin{align}\label{equalityConditionFixation}
\rho_{\mathfrak{s},i} = \rho_{\mathfrak{s}',j} &\iff \frac{\det\mathcal{M}\left(\mathfrak{s},\mathbf{e}_{i}\right)}{\det\mathcal{M}} = \frac{\det\mathcal{M}\left(\mathfrak{s}',\mathbf{e}_{j}\right)}{\det\mathcal{M}} \nonumber \\
&\iff \det\mathcal{M}\left(\mathfrak{s},\mathbf{e}_{i}\right) = \det\mathcal{M}\left(\mathfrak{s}',\mathbf{e}_{j}\right) .
\end{align}
\end{linenomath}
Similarly, by Eq. (\ref{absorptionTimeFormula}),
\begin{linenomath}
\begin{align}\label{equalityConditionAbsorption}
\mathbf{t}_{\mathfrak{s}}=\mathbf{t}_{\mathfrak{s}'} &\iff \sum_{j=K+1}^{\left|\mathcal{S}\right|}\frac{\det\mathcal{M}\left(\mathfrak{s},\mathbf{e}_{j}\right)}{\det\mathcal{M}}=\sum_{j=K+1}^{\left|\mathcal{S}\right|}\frac{\det\mathcal{M}\left(\mathfrak{s}',\mathbf{e}_{j}\right)}{\det\mathcal{M}} \nonumber \\
&\iff \sum_{j=K+1}^{\left|\mathcal{S}\right|}\det\mathcal{M}\left(\mathfrak{s},\mathbf{e}_{j}\right)=\sum_{j=K+1}^{\left|\mathcal{S}\right|}\det\mathcal{M}\left(\mathfrak{s}',\mathbf{e}_{j}\right) .
\end{align}
\end{linenomath}
Assuming (i) and (ii), Eqs. (\ref{equalityConditionFixation}) and (\ref{equalityConditionAbsorption}) are equivalent to polynomial equations in either $\beta$ or $\exp\left\{\beta\right\}$. Either way, since nonzero polynomial equations have at most finitely many solutions, we see that the equalities $\rho_{\mathfrak{s},i}=\rho_{\mathfrak{s}',j}$ and $\mathbf{t}_{\mathfrak{s}}=\mathbf{t}_{\mathfrak{s}'}$ each hold for either (i) every $\beta$ or (ii) finitely many values of $\beta$. Thus, if $\rho_{\mathfrak{s},i}\neq\rho_{\mathfrak{s}',j}$ (resp. $\mathbf{t}_{\mathfrak{s}}\neq\mathbf{t}_{\mathfrak{s}'}$) for even a single selection intensity, then these fixation probabilities (resp. absorption times) differ for almost every selection intensity. In particular, they differ for all sufficiently small $\beta$.
\end{proof}

\setcounter{section}{0}
\setcounter{subsection}{0}
\setcounter{equation}{0}
\renewcommand{\thesubsection}{B.\arabic{subsection}}
\renewcommand{\theequation}{B.\arabic{equation}}

\section*{Appendix B: symmetry and evolutionary equivalence}

\subsection{Symmetries of graphs}\label{subsec:graphSymmetries}

Here we recall some standard notions of symmetry for graphs. Although we treat directed, weighted graphs in general, throughout the main text we give several examples of \textit{un}directed and \textit{un}weighted graphs, which are defined as follows:

\begin{definition}[Undirected graph]
A graph, $\mathscr{D}$, is \textit{undirected} if $\mathscr{D}_{ij}=\mathscr{D}_{ji}$ for each $i$ and $j$.
\end{definition}

\begin{definition}[Unweighted graph]\label{def:unweightedGraph}
A graph, $\mathscr{D}$, is \textit{unweighted} if $\mathscr{D}_{ij}\in\left\{0,1\right\}$ for each $i$ and $j$.
\end{definition}

Since our goal is to discuss symmetry in the context of evolutionary processes, we first describe several notions of symmetry for graphs. In a graph, $\mathscr{D}$, the \textit{indegree} and \textit{outdegree} of vertex $i$ are $\sum_{j=1}^{N}\mathscr{D}_{ji}$ and $\sum_{j=1}^{N}\mathscr{D}_{ij}$, respectively. With these definitions in mind, we recall the definition of a \textit{regular} graph:
\begin{definition}[Regular graph]\label{def:regularGraph}
$\mathscr{D}$ is regular if and only if there exists $k\in\mathbb{R}$ such that
\begin{linenomath}
\begin{align}
\sum_{j=1}^{N}\mathscr{D}_{ji} &= \sum_{j=1}^{N}\mathscr{D}_{ij} = k
\end{align}
\end{linenomath}
for each $i$. If $\mathscr{D}$ is regular, then $k$ is called the \textit{degree} of $\mathscr{D}$.
\end{definition}

Let $\mathfrak{S}_{N}$ denote the symmetric group on $N$ letters; that is, $\mathfrak{S}_{N}$ is the set of all bijections $\pi :\left\{1,\dots ,N\right\}\rightarrow\left\{1,\dots ,N\right\}$. Each $\pi\in\mathfrak{S}_{N}$ extends to a relabeling action on the set of directed, weighted graphs defined by $\left(\pi\mathscr{D}\right)_{ij}=\mathscr{D}_{\pi\left(i\right)\pi\left(j\right)}$. In other words, any relabeling of the set of vertices results in a corresponding relabeling of the graph. The \textit{automorphism group of }$\mathscr{D}$, written $\textrm{Aut}\left(\mathscr{D}\right)$, is the set of all $\pi\in\mathfrak{S}_{N}$ such that $\pi\mathscr{D}=\mathscr{D}$. We now recall a condition slightly stronger than regularity known as vertex-transitivity:
\begin{definition}[Vertex-transitive graph]\label{def:vertexTransitiveGraph}
$\mathscr{D}$ is \textit{vertex-transitive} if for each $i$ and $j$, there exists $\pi\in\textrm{Aut}\left(\mathscr{D}\right)$ such that $\pi\left(i\right) =j$.
\end{definition}
Informally, a graph is vertex-transitive if and only if it ``looks the same" from every vertex. If a graph is vertex-transitive, then it is necessarily regular. The strongest form of symmetry for graphs that we consider here is the following:
\begin{definition}[Symmetric graph]\label{def:symmetricGraph}
$\mathscr{D}$ is \textit{symmetric} (or \textit{arc-transitive}) if for each $i,j$ with $\mathscr{D}_{ij}\neq 0$ and $i',j'$ with $\mathscr{D}_{i'j'}\neq 0$, there exists $\pi\in\textrm{Aut}\left(\mathscr{D}\right)$ such that $\pi\left(i\right) =i'$ and $\pi\left(j\right) =j'$.
\end{definition}
A graph is symmetric if it ``looks the same" from any two directed edges. Arc-transitivity is typically defined for unweighted graphs, i.e. graphs satisfying $\mathscr{D}\in\left\{0,1\right\}^{N\times N}$. For the more general class of weighted graphs, we require that $\mathfrak{S}_{N}$ act transitively on the set of edges of $\mathscr{D}$, where ``edge" means a pair $\left(i,j\right)$ with $\mathscr{D}_{ij}\neq 0$. Thus, all of the edges in a symmetric, weighted graph have the same weight: otherwise, if $\left(i,j\right)$ and $\left(i',j'\right)$ are edges but $\mathscr{D}_{ij}\neq\mathscr{D}_{i'j'}$, then there would exist no permutation, $\pi$, sending $i$ to $i'$, $j$ to $j'$, and preserving the weights of the graph. Therefore, since the weights of a symmetric graph take one of two values ($0$ or else the only nonzero weight), such a graph is essentially unweighted.

\subsection{Symmetries of evolutionary processes}\label{subsec:processSymmetries}

In \S\ref{subsec:generalMarkov}, we defined two states, $\mathfrak{s}$ and $\mathfrak{s}'$, to be \textit{evolutionarily equivalent} if (i) there exists an automorphism of the Markov chain, $\phi\in\textrm{Aut}\left(X\right)$, such that $\phi\left(\mathfrak{s}\right) =\mathfrak{s}'$, and (ii) this automorphism satisfies $\phi\left(\mu\right) =\mu$ for each stationary distribution, $\mu$, of the chain. Condition (i), which means that $\mathfrak{s}$ and $\mathfrak{s}'$ are \textit{symmetric}, alone is not quite strong enough to guarantee that $\mathfrak{s}$ and $\mathfrak{s}'$ have the same long-run behavior. To give an example of a symmetry of states that is not an evolutionary equivalence, we consider the neutral Moran process in a well-mixed population of size $N=3$:
\begin{example}\label{ex:swapStrategies}
In a well-mixed population of size $N=3$, consider the (frequency-independent) Moran process with two types of players: a mutant type and a wild type. Suppose that the mutant type is neutral with respect to the mutant; that is, the fitness of the mutant relative to the wild type is $1$. Since the population is well-mixed, the state of the population is given by the number of mutants it contains, $i\in\left\{0,1,2,3\right\}=:\mathcal{S}$. Consider the map $\phi :\mathcal{S}\rightarrow\mathcal{S}$ defined by $\phi\left(i\right) =3-i$. States $0$ and $3$ are absorbing, and, for $i\in\left\{1,2\right\}$, the transition probabilities of this process are as follows:
\begin{linenomath}
\begin{subequations}
\begin{align}
\mathbf{T}_{i,i-1} &= \mathbf{T}_{i,i+1} = \left(\frac{i}{3}\right)\left(\frac{3-i}{3}\right) ; \\
\mathbf{T}_{i,i} &= \left(\frac{i}{3}\right)^{2}+\left(\frac{3-i}{3}\right)^{2} .
\end{align}
\end{subequations}
\end{linenomath}
It follows at once that $\phi$ preserves these transition probabilities, so $\phi$ is an automorphism of the Markov chain. Let $\rho_{i}$ be the probability that mutants fixate given an initial abundance of $i$ mutants. The states $1$ and $2$ are symmetric since $\phi\left(1\right) =2$, but it is \textit{not} true that $\rho_{1}=\rho_{2}$ since $\rho_{1}=1/3$ and $\rho_{2}=2/3$. The reason for this difference in fixation probabilities is that states $1$ and $2$, although symmetric, are not evolutionary equivalent since $\phi$ swaps the two absorbing states of the process.
\end{example}

In contrast to Example \ref{ex:swapStrategies}, processes with unique stationary distributions have the property that every symmetry of the Markov chain is an evolutionary equivalence (Proposition \ref{prop:ergodicSymmetry}). The following lemma establishes Proposition \ref{prop:ergodicSymmetry}:
\begin{lemma}\label{lem:stationaryDistribution}
If $\phi :\mathcal{S}\rightarrow\mathcal{S}$ is a symmetry of a Markov chain and $\mu$ is a stationary distribution of this chain, then $\phi\left(\mu\right)$ is also a stationary distribution. In particular, if $\mu$ is unique, then $\phi\left(\mu\right) =\mu$.
\end{lemma}
\begin{proof}
If $\mathbf{T}$ is the transition matrix of this Markov chain, then
\begin{linenomath}
\begin{align}
\left[\phi\left(\mu\right)^{T}\mathbf{T}\right]_{\mathfrak{s}} &= \sum_{\mathfrak{s}'\in\mathcal{S}}\phi\left(\mu\right)_{\mathfrak{s}'}\mathbf{T}_{\mathfrak{s}',\mathfrak{s}} \nonumber \\
&= \sum_{\mathfrak{s}'\in\mathcal{S}}\mu_{\phi\left(\mathfrak{s}'\right)}\mathbf{T}_{\phi\left(\mathfrak{s}'\right) ,\phi\left(\mathfrak{s}\right)} \nonumber \\
&= \sum_{\mathfrak{s}'\in\mathcal{S}}\mu_{\mathfrak{s}'}\mathbf{T}_{\mathfrak{s}',\phi\left(\mathfrak{s}\right)} \nonumber \\
&= \left[\mu^{T}\mathbf{T}\right]_{\phi\left(\mathfrak{s}\right)} \nonumber \\
&= \mu_{\phi\left(\mathfrak{s}\right)} \nonumber \\
&= \phi\left(\mu\right)_{\mathfrak{s}} ,
\end{align}
\end{linenomath}
so $\phi\left(\mu\right)^{T}\mathbf{T}=\phi\left(\mu\right)^{T}$, which completes the proof.
\end{proof}

We turn now to the proofs of our main results (Theorems \ref{thm:mainSymmetricTheorem} and \ref{thm:mainAsymmetricTheorem}):
\begin{repeatedSymmetricTheorem}
Consider an evolutionary matrix game on a graph, $\Gamma =\left(\mathscr{E},\mathscr{D}\right)$, with symmetric payoffs and homogeneous strategy mutations. If $\pi\in\textrm{Aut}\left(\Gamma\right)$, then the states with a single mutant at vertex $i$ and $\pi\left(i\right)$, respectively, in an otherwise-monomorphic population, are evolutionarily equivalent. That is, in the notation of Definition \ref{def:homogeneousEvolutionaryProcess}, the states $\mathfrak{s}_{\left(s',i\right) ,s}$ and $\mathfrak{s}_{\left(s',\pi\left(i\right)\right) ,s}$ are evolutionarily equivalent for each $s,s'\in S$.
\end{repeatedSymmetricTheorem}
\begin{proof}
The state space of the Markov chain defined by this evolutionary game is $S^{N}$, where $S$ is the strategy set and $N$ is the population size. For $\pi\in\textrm{Aut}\left(\Gamma\right)$, let $\pi$ act on the state of the population by changing the strategy of player $i$ to that of player $\pi^{-1}\left(i\right)$ for each $i=1,\dots ,N$. In other words, $\pi$ sends $\mathfrak{s}\in S^{N}$ to $\pi\mathfrak{s}\in S^{N}$, which is defined by $\left(\pi\mathfrak{s}\right)_{i}=\mathfrak{s}_{\pi^{-1}\left(i\right)}$ for each $i$. Therefore, for $s,s'\in S$, we have
\begin{linenomath}
\begin{align}
\pi\mathfrak{s}_{\left(s',i\right) ,s} &= \mathfrak{s}_{\left(s',\pi\left(i\right)\right) ,s}
\end{align}
\end{linenomath}
for each $i=1,\dots ,N$. Since $\pi$ is an automorphism of the evolutionary graph, $\Gamma$, we have $\pi\mathscr{E}=\mathscr{E}$ and $\pi\mathscr{D}=\mathscr{D}$. Moreover, $\pi$ preserves the strategy mutations since they are homogeneous. Since the payoffs are symmetric, $\pi$ just rearranges the fitness profile of the population: the payoff of player $i$ becomes the payoff of player $\pi^{-1}\left(i\right)$ (see Eq. (\ref{eq:symmetricMatrixPayoff})), so the same is true of the fitness values. Therefore, applying the map $\pi$ to $S^{N}$ is equivalent to applying the map on $S^{N}$ obtained by simply relabeling the players. Since any such relabeling of the players results in an automorphism of the Markov chain on $S^{N}$ that preserves the monomorphic absorbing states, it follows that $\mathfrak{s}_{\left(s',i\right) ,s}$ and $\mathfrak{s}_{\left(s',\pi\left(i\right)\right) ,s}$ are evolutionarily equivalent.
\end{proof}

\begin{repeatedAsymmetricTheorem}
Suppose that an asymmetric matrix game with homogeneous strategy mutations is played on an evolutionary graph, $\Gamma =\left(\mathscr{E},\mathscr{D}\right)$. For each $\pi\in\textrm{Aut}\left(\Gamma\right)$, $k\in\left\{1,2,3\right\}$, and $i,j\in\left\{1,\dots ,N\right\}$,
\begin{linenomath}
\begin{align}
\sigma_{k}^{ij} &= \sigma_{k}^{\pi\left(i\right)\pi\left(j\right)} .
\end{align}
\end{linenomath}
\end{repeatedAsymmetricTheorem}
\begin{proof}
Let $\mathbf{T}$ be the transition matrix for the Markov chain defined by this process. Since there are nonzero strategy mutations, this chain has a unique stationary distribution, $\mu$. The matrix $\mathbf{T}$ defines a directed, weighted graph on $\left|\mathcal{S}\right| =\left|S\right|^{N}$ vertices that has an edge from vertex $\mathfrak{s}$ to vertex $\mathfrak{s}'$ if and only if $\mathbf{T}_{\mathfrak{s},\mathfrak{s}'}\neq 0$. If there is an edge from $\mathfrak{s}$ to $\mathfrak{s}'$, then the weight of this edge is simply $\mathbf{T}_{\mathfrak{s},\mathfrak{s}'}$. The (outdegree) Laplacian matrix of this graph, $\mathcal{L}=\mathcal{L}\left(\mathbf{T}\right)$, is defined by $\mathcal{L}=I_{\left|\mathcal{S}\right|}-\mathbf{T}$ \citep[see][]{chung:AMS:1996}. In terms of this Laplacian matrix, \citet{press:PNAS:2012} show that for any vector, $\nu$, the stationary distribution satisfies
\begin{linenomath}
\begin{align}
\mu\cdot\nu &= \frac{\det\mathcal{L}\left(\mathfrak{s},\nu\right)}{\det\mathcal{L}\left(\mathfrak{s},\mathbf{1}\right)} ,
\end{align}
\end{linenomath}
for each state, $\mathfrak{s}$, where $\mathcal{L}\left(\mathfrak{s},\nu\right)$ denotes the matrix obtained from $\mathcal{L}$ by replacing the column corresponding to state $\mathfrak{s}$ by $\nu$. Thus, if $\psi_{r}$ is the vector indexed by $\mathcal{S}=S^{N}$ with $\psi_{r}\left(\mathfrak{s}\right)$ being the frequency of strategy $r$ in state $\mathfrak{s}$, then the \textit{average abundance} of strategy $r$ is
\begin{linenomath}
\begin{align}\label{eq:averageAbundance}
F_{r} &:= \mu\cdot\psi_{r} = \frac{\det\mathcal{L}\left(\mathfrak{s},\psi_{r}\right)}{\det\mathcal{L}\left(\mathfrak{s},\mathbf{1}\right)} .
\end{align}
\end{linenomath}
Since $\mathbf{T}$ is a function of the payoffs, $\mathbf{a}=\left(a_{st}^{ij}\right)_{s,t,i,j}$, we may write $F_{r}=F_{r}\left(\mathbf{a}\right)$. ($\mathbf{a}$ is just an ordered tuple defined by $\mathbf{a}_{st}^{ij}:=a_{st}^{ij}$ for each $s$, $t$, $i$, and $j$.) Moreover, since the entries of $\mathbf{T}$ are assumed to be smooth functions of $\mathbf{a}$ \citep[see][]{tarnita:PNAS:2011}, $F_{r}$ is also a smooth function of $\mathbf{a}$ by Eq. (\ref{eq:averageAbundance}) and the definition of $\mathcal{L}$. We will show that for each $s$ and $t$,
\begin{linenomath}
\begin{align}\label{eq:partialsEquation}
\frac{\partial F_{r}}{\partial a_{st}^{ij}}\Bigg\vert_{\mathbf{a}=0} &= \frac{\partial F_{r}}{\partial a_{st}^{\pi\left(i\right)\pi\left(j\right)}}\Bigg\vert_{\mathbf{a}=0}
\end{align}
\end{linenomath}
for each $i$ and $j$. The theorem will then follow from the derivations of $\sigma_{1}^{ij}$, $\sigma_{2}^{ij}$, and $\sigma_{3}^{ij}$ in \citep{mcavoy:JMB:2015} since it is shown there that each $\sigma^{ij}$ is a function of the elements in the set
\begin{linenomath}
\begin{align}
\left\{ \frac{\partial F_{r}}{\partial a_{st}^{ij}}\Bigg\vert_{\mathbf{a}=0} \right\}_{s,t=1}^{n} .
\end{align}
\end{linenomath}
For $s$, $t$, $i$, and $j$ fixed and $a\in\mathbb{R}$, let $F_{r}^{s,t,i,j}\left(a\right)$ be the function of the vector with $a$ at entry $a_{st}^{ij}$ and $0$ in all other entries. Symbolically, if $\delta_{x,y}$ is defined as $1$ if $x=y$ and $0$ otherwise and
\begin{linenomath}
\begin{align}
\mathbf{a}_{0}^{\left(s,t,i,j\right)}\left(a\right) &:= \left(a\delta_{s,s'}\delta_{t,t'}\delta_{i,i'}\delta_{j,j'}\right)_{s',t',i',j'} ,
\end{align}
\end{linenomath}
then
\begin{linenomath}
\begin{align}
F_{r}^{s,t,i,j}\left(a\right) &:= F_{r}\left( \mathbf{a}_{0}^{\left(s,t,i,j\right)}\left(a\right) \right) .
\end{align}
\end{linenomath}
Let $\pi\in\mathfrak{S}_{N}$ and suppose that $\pi\in\textrm{Aut}\left(\Gamma\right)$; that is, $\pi\mathscr{E}=\mathscr{E}$ and $\pi\mathscr{D}=\mathscr{D}$. $\pi$ induces a map on the payoffs, $\mathbf{a}$, defined by $\pi\left(a_{st}^{ij}\right)_{s,t,i,j}=\left(a_{st}^{\pi\left(i\right)\pi\left(j\right)}\right)_{s,t,i,j}$. Let $\textrm{orb}_{\mathfrak{S}_{N}}\left(\mathbf{a}\right)$ denote the orbit of $\mathbf{a}$ under this action, and consider the enlarged state space $\mathcal{S}':=S^{N}\times\textrm{orb}_{\mathfrak{S}_{N}}\left(\mathbf{a}\right)$. Using the Markov chain on $S^{N}$ coming from the evolutionary process, we obtain a Markov chain on $S^{N}\times\textrm{orb}_{\mathfrak{S}_{N}}\left(\mathbf{a}\right)$ via the transition matrix, $\mathbf{T}'$, defined by
\begin{linenomath}
\begin{align}
\mathbf{T}_{\left(\mathfrak{s},\mathbf{b}\right) , \left(\mathfrak{s}',\mathbf{b}'\right)}' &:= \delta_{\mathbf{b},\mathbf{b}'}\mathbf{T}_{\mathfrak{s},\mathfrak{s}'}\left(\mathbf{b}\right) .
\end{align}
\end{linenomath}
for $\mathfrak{s},\mathfrak{s}'\in S^{N}$ and $\mathbf{b},\mathbf{b}'\in\textrm{orb}_{\mathfrak{S}_{N}}\left(\mathbf{a}\right)$. (We write $\mathbf{T}\left(\mathbf{b}\right)$ to indicate the transition matrix as a function of the payoff values of the game.) $\pi$ extends to a map on $\mathcal{S}'$ defined by $\pi\left(\mathfrak{s},\mathbf{b}\right) =\left(\pi\mathfrak{s},\pi\mathbf{b}\right)$. Since $\pi$ preserves $\mathscr{E}$, $\mathscr{D}$, and the strategy mutations (since they are homogeneous), it follows that the induced map $\pi :\mathcal{S}'\rightarrow\mathcal{S}'$ is an automorphism of the Markov chain on $\mathcal{S}'$ defined by $\mathbf{T}'$. If $\mu '$ is a stationary distribution for the chain $\mathbf{T}'$, then, for each $\mathbf{s}\in S^{N}$ and $\mathbf{b}\in\textrm{orb}_{\mathfrak{S}_{N}}\left(\mathbf{a}\right)$,
\begin{linenomath}
\begin{align}
\mu_{\left(\mathfrak{s},\mathbf{b}\right)}' &= \sum_{\mathfrak{s}'\in S^{N}}\sum_{\mathbf{b}'\in\textrm{orb}_{\mathfrak{S}_{N}}\left(\mathbf{a}\right)}\mu_{\left(\mathfrak{s}',\mathbf{b}'\right)}'\mathbf{T}_{\left(\mathfrak{s}',\mathbf{b}'\right) ,\left(\mathfrak{s},\mathbf{b}\right)}' \nonumber \\
&= \sum_{\mathfrak{s}'\in S^{N}}\sum_{\mathbf{b}'\in\textrm{orb}_{\mathfrak{S}_{N}}\left(\mathbf{a}\right)}\mu_{\left(\mathfrak{s}',\mathbf{b}'\right)}'\delta_{\mathbf{b}',\mathbf{b}}\mathbf{T}_{\mathfrak{s}',\mathfrak{s}} \nonumber \\
&= \sum_{\mathfrak{s}'\in S^{N}}\mu_{\left(\mathfrak{s}',\mathbf{b}\right)}'\mathbf{T}_{\mathfrak{s}',\mathfrak{s}} .
\end{align}
\end{linenomath}
It then follows from the uniqueness of $\mu =\mu\left(\mathbf{b}\right)$ that there exists $c_{\mu '}\left(\mathbf{b}\right)\geqslant 0$ such that $\mu_{\left(\mathfrak{s},\mathbf{b}\right)}'=c_{\mu '}\left(\mathbf{b}\right)\mu_{\mathfrak{s}}$. If $\mu '$ is such a stationary distribution, then, by Lemma \ref{lem:stationaryDistribution}, $\pi\mu' =\mu''$ for some other stationary distribution, $\mu ''$, of the chain on $\mathcal{S}'$. This equation implies that $c_{\mu ''}\left(\mathbf{b}\right) =c_{\mu '}\left(\pi\mathbf{b}\right)$ for each $\mathbf{b}\in\textrm{orb}_{\mathfrak{S}_{N}}\left(\mathbf{a}\right)$. Therefore,
\begin{linenomath}
\begin{align}
\pi\mu\left(\pi\mathbf{b}\right)_{\mathfrak{s}} &= \mu\left(\pi\mathbf{b}\right)_{\pi\mathfrak{s}} = \mu\left(\mathbf{b}\right)_{\mathfrak{s}}
\end{align}
\end{linenomath}
for each $\mathfrak{s}\in S^{N}$ and $\mathbf{b}\in\textrm{orb}_{\mathfrak{S}_{N}}\left(\mathbf{a}\right)$. Consequently, since $\pi\psi_{r}=\psi_{r}$,
\begin{linenomath}
\begin{align}
F_{r}^{s,t,i,j}\left(a\right) &= F_{r}\left(\mathbf{a}_{0}^{\left(s,t,i,j\right)}\left(a\right)\right) \nonumber \\
&= \mu\left(\mathbf{a}_{0}^{\left(s,t,i,j\right)}\left(a\right)\right)\cdot\psi_{r} \nonumber \\
&= \pi\mu\left(\pi\mathbf{a}_{0}^{\left(s,t,i,j\right)}\left(a\right)\right)\cdot\psi_{r} \nonumber \\
&= \pi\mu\left(\mathbf{a}_{0}^{\left(s,t,\pi\left(i\right) ,\pi\left(j\right)\right)}\left(a\right)\right)\cdot\pi\psi_{r} \nonumber \\
&= \mu\left(\mathbf{a}_{0}^{\left(s,t,\pi\left(i\right) ,\pi\left(j\right)\right)}\left(a\right)\right)\cdot\psi_{r} \nonumber \\
&= F_{r}\left(\mathbf{a}_{0}^{\left(s,t,\pi\left(i\right) ,\pi\left(j\right)\right)}\left(a\right)\right) \nonumber \\
&= F_{r}^{s,t,\pi\left(i\right) ,\pi\left(j\right)}\left(a\right) .
\end{align}
\end{linenomath}
As a result, we have
\begin{linenomath}
\begin{align}
\frac{\partial F_{r}}{\partial a_{st}^{ij}}\Bigg\vert_{\mathbf{a}=0} &= \frac{d}{da}\Bigg\vert_{a=0}F_{r}^{s,t,i,j} = \frac{d}{da}\Bigg\vert_{a=0}F_{r}^{s,t,\pi\left(i\right) ,\pi\left(j\right)} = \frac{\partial F_{r}}{\partial a_{st}^{\pi\left(i\right)\pi\left(j\right)}}\Bigg\vert_{\mathbf{a}=0} ,
\end{align}
\end{linenomath}
so Eq. (\ref{eq:partialsEquation}) holds, which completes the proof.
\end{proof}

\setcounter{section}{0}
\setcounter{subsection}{0}
\setcounter{equation}{0}
\renewcommand{\thesubsection}{C.\arabic{subsection}}
\renewcommand{\theequation}{C.\arabic{equation}}

\section*{Appendix C: explicit calculations}

We now perform explicit calculations using Eqs. (\ref{fixationProbabilityFormula}) and (\ref{absorptionTimeFormula}) to show that the Isothermal Theorem extends to neither absorption times nor frequency-dependent games. We also calculate the structure coefficients for the death-birth process on the graph of Fig. \ref{fig:transitiveSymmetric}(A) to show that Corollary \ref{cor:spatialAverage} does not necessarily hold for graphs that are vertex-transitive but not symmetric.

\subsection{The Moran process}\label{subsec:moranAppendix}

Consider the Moran process on the Frucht graph (Fig. \ref{fig:fruchtGraph}), and suppose that the mutant type has fitness $r>0$ relative to the wild type. By the Isothermal Theorem of \citet{lieberman:Nature:2005}, the fixation probability of a fixed number of mutants is independent of the configuration of those mutants on the graph. For $r=2$, the absorption times (of configurations of a single mutant in a wild-type population) are listed in Table \ref{table:absorption}. The fixation probability of a single mutant is $\left(1-\frac{1}{2}\right) /\left(1-\frac{1}{2^{12}}\right)\approx 0.5001$ for every vertex. Thus, unlike fixation probabilities, absorption times depend on the initial location of the mutant. (Some of the absorption times are similar in this case, but no two are the same.)
\begin{table}
\begin{tabular}{|c|c|c|}
\hline
\textbf{initial vertex of mutant} & \textbf{absorption time} \\
\hline
1 & 238.1836 \\
\hline
2 & 237.0596 \\
\hline
3 & 234.5982 \\
\hline
4 & 235.8447 \\
\hline
5 & 236.5967 \\
\hline
6 & 234.5792 \\
\hline
7 & 231.6988 \\
\hline
8 & 238.0375 \\
\hline
9 & 233.6122 \\
\hline
10 & 235.1514 \\
\hline
11 & 230.1340 \\
\hline
12 & 228.7114 \\
\hline
\end{tabular}
\caption{The absorption times of the $12$ initial configurations of a single mutant in a wild-type population for the Moran process on the Frucht graph. The fitness of the mutant relative to the wild type is $r=2$.\label{table:absorption}}
\end{table}

\subsection{Frequency-dependent games}

\subsubsection{Symmetric games}\label{subsubsec:symmetricAppendix}

Consider the instance of the Snowdrift Game that has for a payoff matrix (\ref{snowdriftGame}). For this game, Table \ref{table:fixationAbsorption} gives the fixation probabilities and the absorption times (rounded to four digits after the decimal point) for the death-birth process on the Frucht graph (Fig. \ref{fig:fruchtGraph}) with $\beta =1$.
\begin{table}
\begin{tabular}{|c|c|c|}
\hline
\textbf{initial vertex of mutant} & \textbf{fixation probability} & \textbf{absorption time} \\
\hline
1 & 0.6505 & 116.0959 \\
\hline
2 & 0.6471 & 115.4026 \\
\hline
3 & 0.6469 & 115.7302 \\
\hline
4 & 0.6448 & 115.7348 \\
\hline
5 & 0.6463 & 116.0100 \\
\hline
6 & 0.6562 & 117.4671 \\
\hline
7 & 0.7299 & 129.8609 \\
\hline
8 & 0.6512 & 116.6906 \\
\hline
9 & 0.6545 & 118.3795 \\
\hline
10 & 0.6551 & 117.8995 \\
\hline
11 & 0.7344 & 131.6681 \\
\hline
12 & 0.7326 & 131.9634 \\
\hline
\end{tabular}
\caption{The fixation probabilities and absorption times of the $12$ initial configurations of a single cooperator among defectors for the death-birth process on the Frucht graph. Payoffs are frequency-dependent and derived from the Snowdrift Game, (\ref{snowdriftGame}). The intensity of selection is $\beta =1$.\label{table:fixationAbsorption}}
\end{table}

Similarly, for the same game (and update rule) but on the Tietze graph (Fig. \ref{fig:tietze}) with $\beta =0.1$, Table \ref{table:fixationAbsorptionTietze} and Fig. \ref{fig:tietzeBarGraphs} give the fixation probabilities and absorption times for all possible configurations of a single cooperator among defectors.

\begin{table}
\begin{tabular}{|c|c|c|}
\hline
\textbf{initial vertex of mutant} & \textbf{fixation probability} & \textbf{absorption time} \\
\hline
1 & 0.3777 & 70.7869 \\
\hline
2 & 0.3777 & 70.7869 \\
\hline
3 & 0.3777 & 70.7869 \\
\hline
4 & 0.4141 & 76.5048 \\
\hline
5 & 0.4186 & 77.3094 \\
\hline
6 & 0.4186 & 77.3094 \\
\hline
7 & 0.4141 & 76.5048 \\
\hline
8 & 0.4186 & 77.3094 \\
\hline
9 & 0.4186 & 77.3094 \\
\hline
10 & 0.4141 & 76.5048 \\
\hline
11 & 0.4186 & 77.3094 \\
\hline
12 & 0.4186 & 77.3094 \\
\hline
\end{tabular}
\caption{The fixation probabilities and absorption times of the $12$ initial configurations of a single cooperator among defectors for the death-birth process on the Tietze graph. Payoffs are frequency-dependent and derived from the Snowdrift Game, (\ref{snowdriftGame}). The intensity of selection is $\beta =0.1$. These values are illustrated graphically in Fig. \ref{fig:tietzeBarGraphs}.\label{table:fixationAbsorptionTietze}}
\end{table}

\begin{figure}
\begin{center}
\subfloat[]{\includegraphics[scale=0.45]{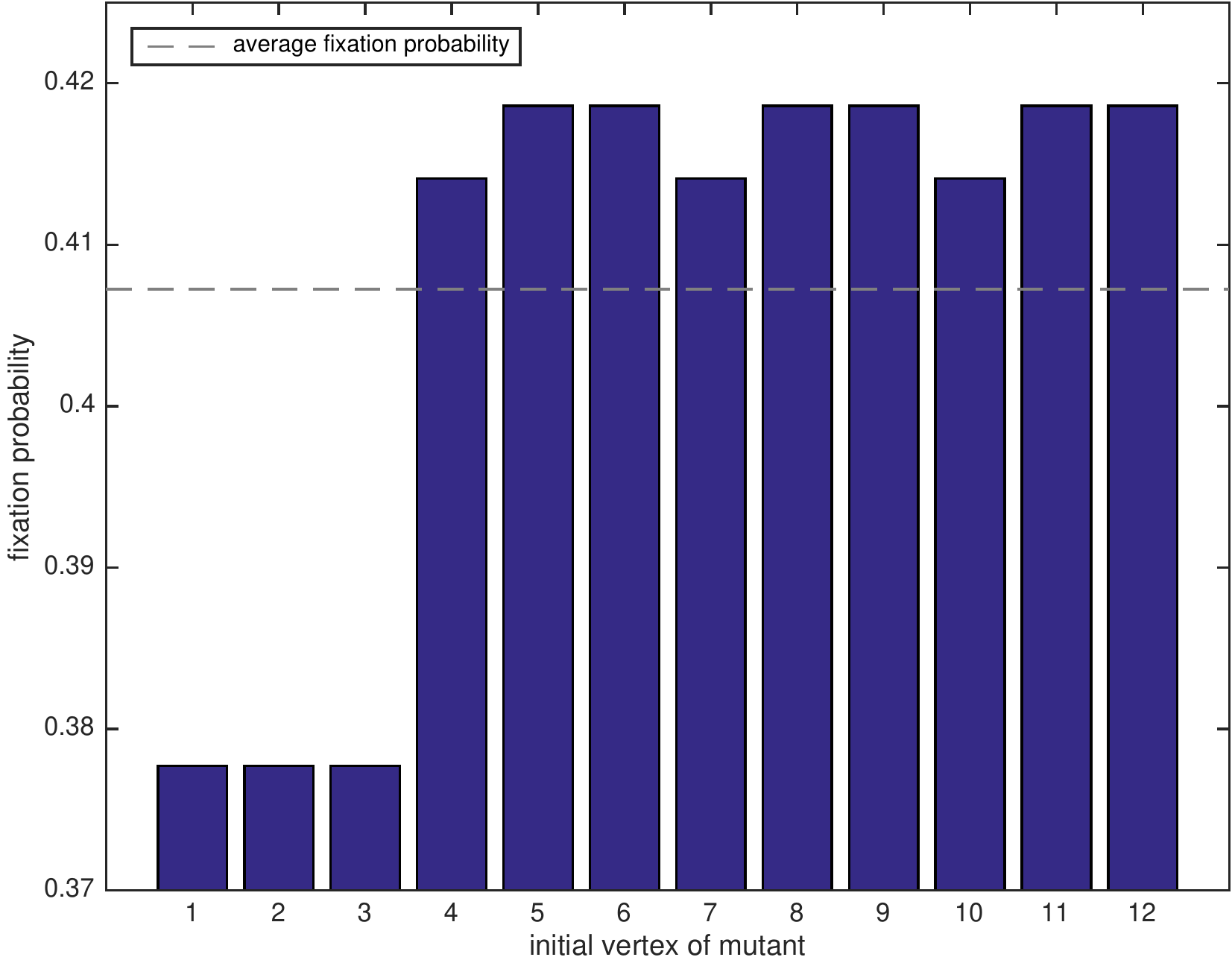}}
\quad\quad
\subfloat[]{\includegraphics[scale=0.45]{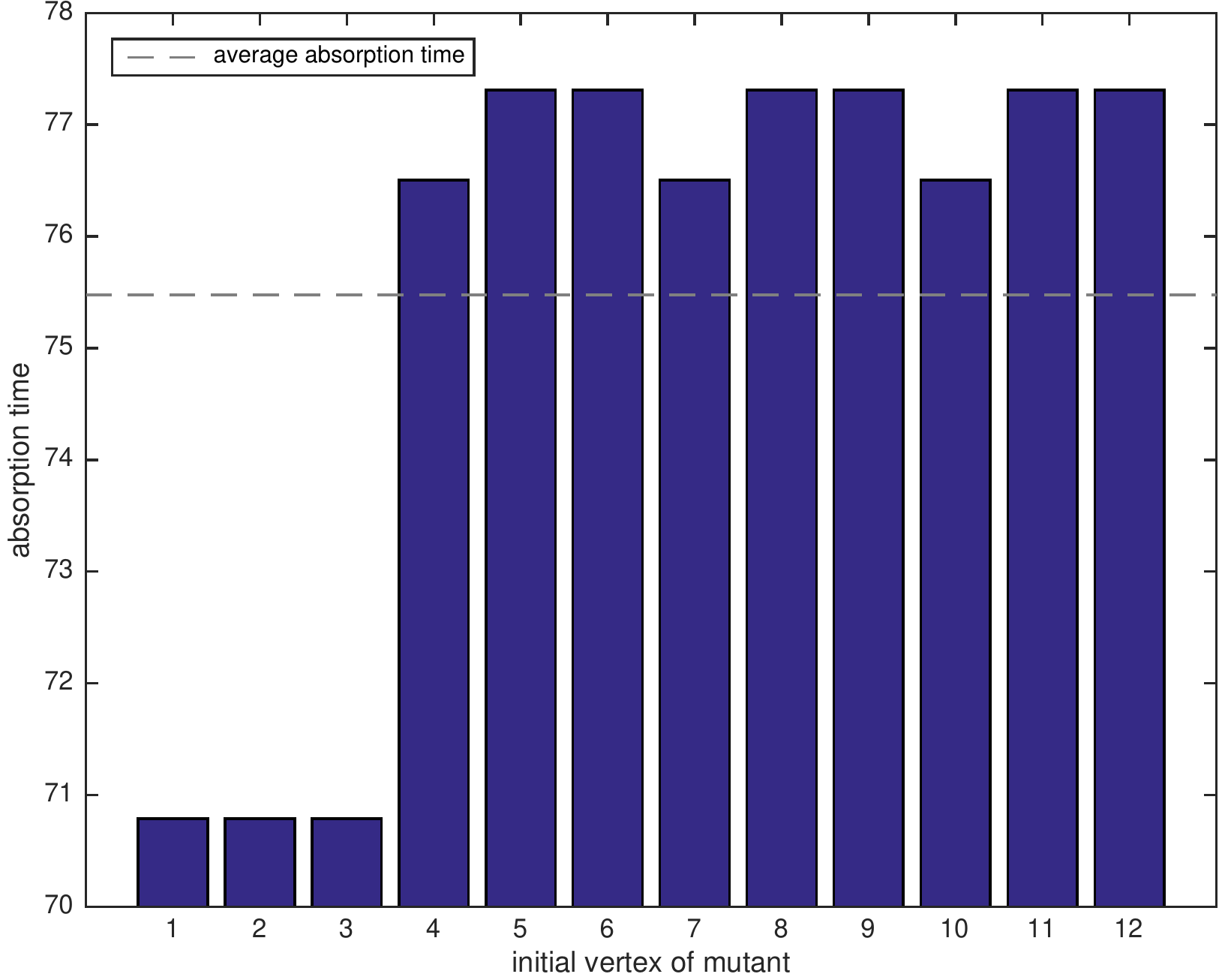}}
\end{center}
\caption{Fixation probability (A) and absorption time (B) versus initial vertex of mutant (cooperator) for a death-birth process on the Tietze graph. In both figures, the game is a Snowdrift Game whose payoffs are described by payoff matrix (\ref{snowdriftGame}), and the selection intensity is $\beta =0.1$. This example illustrates the single-mutant states that are not evolutionarily equivalent in the Tietze graph. Moreover, it happens to be the case that any two of these states with the same fixation probability (or absorption time) are evolutionarily equivalent.\label{fig:tietzeBarGraphs}}
\end{figure}

\subsubsection{Asymmetric games}\label{subsubsection:asymmetricAppendix}

For the death-birth process on the graph in Fig. \ref{fig:transitiveSymmetric}(A) with homogeneous strategy-mutation rate $\varepsilon =0.01$, we calculate the complete collection of structure coefficients $\left\{\tau_{1}^{ij},\tau_{2}^{ij}\right\}_{i,j}$ (for $r=1$) as follows: Let $\psi_{1}$ be the vector indexed by $\mathcal{S}$ with $\psi_{1}\left(\mathfrak{s}\right)$ being the frequency of strategy $1$ in state $\mathfrak{s}$, and let $\mathbf{1}$ be the vector of ones. \citet{mcavoy:JMB:2015} show that, for any $\mathfrak{s}$, Eq. (\ref{asymmetricNetwork}) is equivalent to
\begin{linenomath}
\begin{align}
\frac{1}{n}\textrm{tr}\left(\left(\Big[\mathcal{L}\left(\mathfrak{s},\psi_{1}\right)\vert_{\beta =0}\Big]^{-1} - \Big[\mathcal{L}\left(\mathfrak{s},\mathbf{1}\right)\vert_{\beta =0}\Big]^{-1}\right)\frac{d}{d\beta}\Big\vert_{\beta =0}\mathcal{L}\left(\mathfrak{s},\mathbf{0}\right)\right) &> 0 ,
\end{align}
\end{linenomath}
where $\mathcal{L}=\mathcal{L}\left(\mathbf{T}\right) =I_{\left|\mathcal{S}\right|}-\mathbf{T}$ is the (outdegree) Laplacian matrix of the graph defined by $\mathbf{T}$. If $k\in\left\{1,2\right\}$, $i$ and $j$ are fixed, and we choose the (two-strategy) asymmetric game so that
\begin{linenomath}
\begin{align}
a_{st}^{i'j'} &= \begin{cases}1 & s=1,\ t=k,\ i'=i,\ j'=j; \\ 0 & \textrm{otherwise},\end{cases}
\end{align}
\end{linenomath}
then
\begin{linenomath}
\begin{align}
\tau_{k}^{ij} &= \frac{1}{2}\textrm{tr}\left(\left(\Big[\mathcal{L}\left(\mathfrak{s},\psi_{1}\right)\vert_{\beta =0}\Big]^{-1} - \Big[\mathcal{L}\left(\mathfrak{s},\mathbf{1}\right)\vert_{\beta =0}\Big]^{-1}\right)\frac{d}{d\beta}\Big\vert_{\beta =0}\mathcal{L}\left(\mathfrak{s},\mathbf{0}\right)\right) .
\end{align}
\end{linenomath}
Using this method, we obtain the structure coefficients for Fig. \ref{fig:transitiveSymmetric}(A) listed in Table \ref{table:structureCoefficients}. For the same process on the symmetric graph of Fig. \ref{fig:transitiveSymmetric}(B), we find that $\tau_{1}^{ij}$ and $\tau_{2}^{ij}$ are independent of $i$ and $j$ and are both equal to $2189/27728$.
\begin{table}
\begin{tabular}{|c|c|c|}
\hline
$\left(i,j\right)$ & $\tau_{1}^{ij}$ & $\tau_{2}^{ij}$ \\
\hline
$\left(1,2\right)$ & $707905/9315552$ & $32989/405024$ \\
\hline
$\left(1,4\right)$ & $16291/194074$ & $57057/776296$ \\
\hline
$\left(1,6\right)$ & $707905/9315552$ & $32989/405024$ \\
\hline
$\left(2,1\right)$ & $707905/9315552$ & $32989/405024$ \\
\hline
$\left(2,3\right)$ & $16291/194074$ & $57057/776296$ \\
\hline
$\left(2,6\right)$ & $707905/9315552$ & $32989/405024$ \\
\hline
$\left(3,2\right)$ & $16291/194074$ & $57057/776296$ \\
\hline
$\left(3,4\right)$ & $707905/9315552$ & $32989/405024$ \\
\hline
$\left(3,5\right)$ & $707905/9315552$ & $32989/405024$ \\
\hline
$\left(4,1\right)$ & $16291/194074$ & $57057/776296$ \\
\hline
$\left(4,3\right)$ & $707905/9315552$ & $32989/405024$ \\
\hline
$\left(4,5\right)$ & $707905/9315552$ & $32989/405024$ \\
\hline
$\left(5,3\right)$ & $707905/9315552$ & $32989/405024$ \\
\hline
$\left(5,4\right)$ & $707905/9315552$ & $32989/405024$ \\
\hline
$\left(5,6\right)$ & $16291/194074$ & $57057/776296$ \\
\hline
$\left(6,1\right)$ & $707905/9315552$ & $32989/405024$ \\
\hline
$\left(6,2\right)$ & $707905/9315552$ & $32989/405024$ \\
\hline
$\left(6,5\right)$ & $16291/194074$ & $57057/776296$ \\
\hline
\end{tabular}
\caption{The structure coefficients in Eq. (\ref{asymmetricNetworkTwo}) for the death-birth process on the vertex-transitive (but not symmetric) graph of Fig. \ref{fig:transitiveSymmetric}(A) with homogeneous strategy-mutation rate $\varepsilon =0.01$.\label{table:structureCoefficients}}
\end{table}

\end{document}